\documentclass[aps,showpacs,amssymb,amsfonts,superscriptaddress,twocolumn,pra]{revtex4-1}

\usepackage{amsmath,bbm,mathrsfs}
\usepackage{amsthm}
\usepackage{graphicx}
\usepackage{amsfonts}
\usepackage{amssymb}
\usepackage{color}
\usepackage{enumitem}

\usepackage{changes}

\definecolor{myurlcolor}{rgb}{0,0,0.7}
\definecolor{myrefcolor}{rgb}{0.8,0,0}
\usepackage{hyperref}
\hypersetup{colorlinks,
linkcolor=myrefcolor,
citecolor=myurlcolor,
urlcolor=myurlcolor}

\font\Bbb =msbm10 
 \def\C{{\hbox {\Bbb C}}}

\def\textbf#1{{\bf #1}}
\def\be{\begin{equation}}
\def\ee{\end{equation}}
\def\ben{\begin{eqnarray}}
\def\een{\end{eqnarray}}
\def\eea{\end{array}}
\def\bea{\begin{array}}
\newcommand{\ot}[0]{\otimes}
\newcommand{\Tr}[1]{\mathrm{Tr}#1}
\newcommand{\bei}{\begin{itemize}}
\newcommand{\eei}{\end{itemize}}
\newcommand{\ket}[1]{|#1\rangle}

\begin{document}

\newcommand{\eg}{{\it{e.g.~}}}
\newcommand{\ie}{{\it{i.e.~}}}
\newcommand{\etal}{{\it{et al.}}}
\newcommand{\daniel}[1]{{\color{red} #1}}
\newcommand{\remik}[1]{{\color{green} #1}}

\newtheorem{thm}{Theorem}
\newtheorem{lem}[thm]{Lemma}
\theoremstyle{remark}
\newtheorem{example}{Example}
\newtheorem{cor}[thm]{Corollary}
\newtheorem{rem}[thm]{Remark}
\newcommand{\R}{\mathbb{R}}
\newcommand{\beq}{\begin{equation}}
\newcommand{\eeq}{\end{equation}}

\title{Tight Bell inequalities with no quantum violation from qubit unextendible product bases}

\author{Remigiusz Augusiak}
\author{Tobias Fritz}

\affiliation{ICFO--Institut de Ci\`{e}ncies Fot\`{o}niques, 08860
Castelldefels (Barcelona), Spain}

\author{Marcin Kotowski}
\author{Micha\l{} Kotowski}

\affiliation{Department of Combinatorics and Optimization and
Institute for Quantum Computing University of Waterloo, 200
University Avenue, Waterloo, Ontario,  N2L 3G1, Canada}

\author{Marcin Paw\l{}owski}

\affiliation{Department of Mathematics, University of Bristol, Bristol BS8 1TW, United Kingdom}

\author{Maciej Lewenstein}

\author{Antonio Ac\'in}

\affiliation{ICFO--Institut de Ci\`{e}ncies Fot\`{o}niques, 08860
Castelldefels (Barcelona), Spain}

\affiliation{ICREA--Instituci\'o Catalana de Recerca i Estudis
Avan\c{c}ats, Lluis Companys 23, 08010 Barcelona, Spain}

\begin{abstract}
We investigate the relation between unextendible product bases
(UPB) and Bell inequalities found recently in [R. Augusiak {\it et
al.}, \href{http://prl.aps.org/abstract/PRL/v107/i7/e070401}{Phys.
Rev. Lett. {\bf 107}, 070401 (2011)}]. We first review the
procedure introduced there that associates to any set of mutually
orthogonal product vectors in a many-qubit Hilbert space a Bell
inequality. We then show that if a set of mutually orthogonal
product vectors can be completed to a full basis, then the
associated Bell inequality is trivial, in the sense of not being
violated by any nonsignalling correlations. This implies that the
relevant Bell inequalities that arise from the construction all
come from UPBs, which adds additional weight to the significance
of UPBs for Bell inequalities. Then, we provide new examples of
{\it tight} Bell inequalities with no quantum violation
constructed from UPBs in this way. Finally, it is proven that the
Bell inequalities with no quantum violation introduced recently in
[M. Almeida {\it et al.},
\href{http://prl.aps.org/abstract/PRL/v104/i23/e230404}{Phys. Rev.
Lett. {\bf 104}, 230404 (2010)}] are tight for any odd number of
parties.
\end{abstract}

\maketitle

\section{Introduction}

It is well-established that quantum correlations (QC), i.e.,
correlations that can be obtained by local measurements on quantum
states, offer applications with no classical analog. For instance,
they provide cryptographic security not achievable by any
classical cryptographic protocol~\cite{Ekert,bhk,DIQKD}, they
enable the certification of the presence of
randomness~\cite{randomness,colbeck}, and, last but not least,
outperform classical correlations (CC) at communication complexity
tasks (see e.g. Ref.~\cite{Commcompl}).

It is then interesting to ask whether quantum correlations are
always more powerful than classical correlations. In other words,
is it possible to find tasks at which classical correlations
perform equally well as quantum? Such instances can be identified
with the aid of Bell inequalities~\cite{Bell}, which are
constraints satisfied by all CC. Any Bell inequality can be
interpreted as the success probability of a task in which distant
non-communicating parties are each given a certain input and then
must compute, in a distributed manner, the correct value of a
certain known function of the inputs. The violation of a Bell
inequality by some correlations indicates that the corresponding
task can be performed more efficiently by these correlations than
by any CC. Consequently, correlations leading to a Bell violation
do not have a classical realization. On the other hand, Bell
inequalities with no quantum violation provide tasks at which QC
offer no advantage over CC.

The first examples of Bell inequalities that cannot be violated by
quantum theory were derived in Ref.~\cite{Linden}. These
inequalities are nontrivial, as they are violated by
non-signalling correlations, which, necessarily, do not have a
quantum realization. The set of non-signalling correlations (NC)
is defined to be the set of all those correlations which do not
allow any instantaneous communication. However, the Bell
inequalities found in~\cite{Linden} are not tight, which is an
important feature in the present context (see
Figure~\ref{fig:tight}).
\begin{figure}[t]
a)\includegraphics[width=0.22\textwidth]{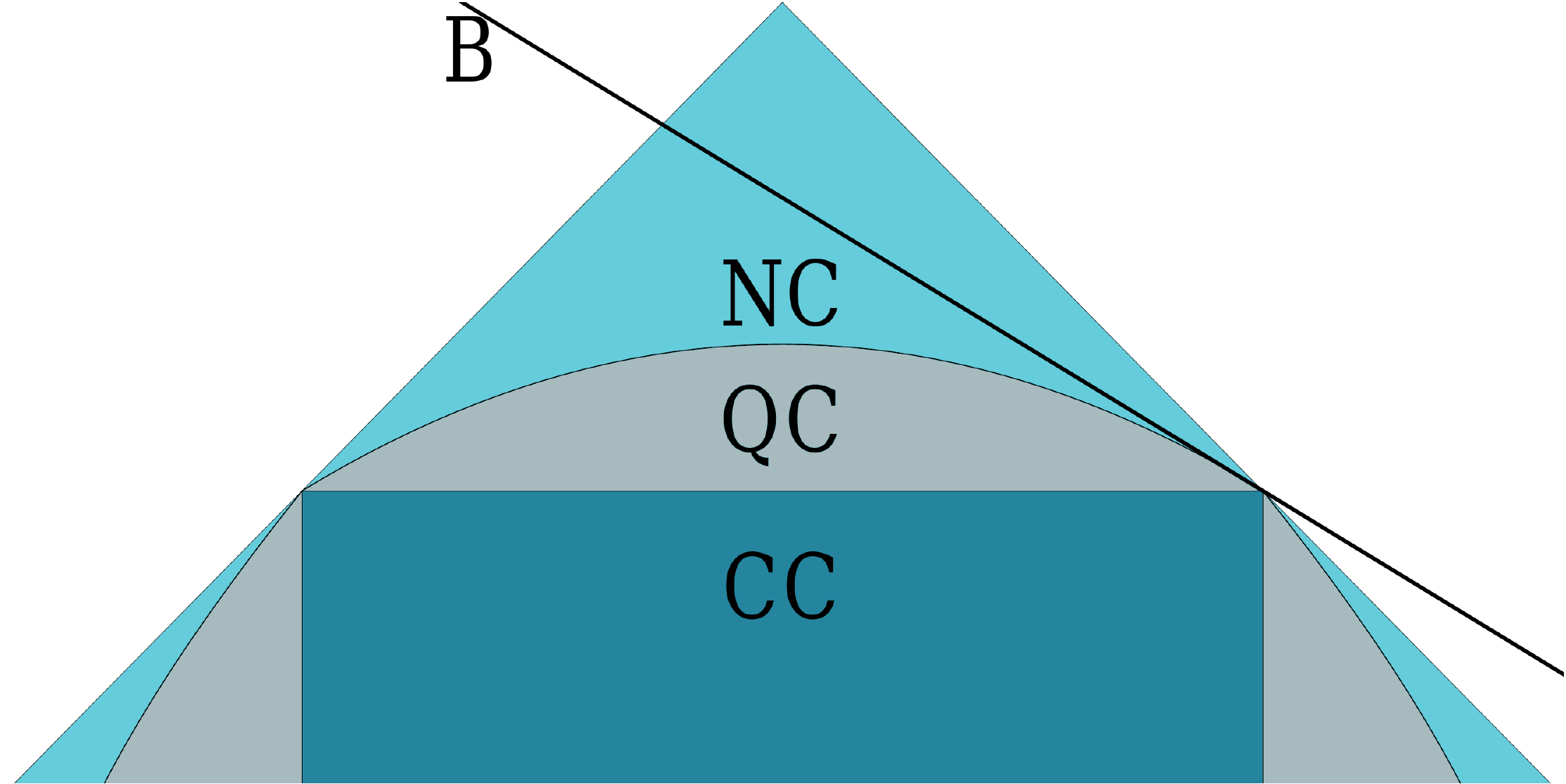}\,\,b)\includegraphics[width=0.22\textwidth]{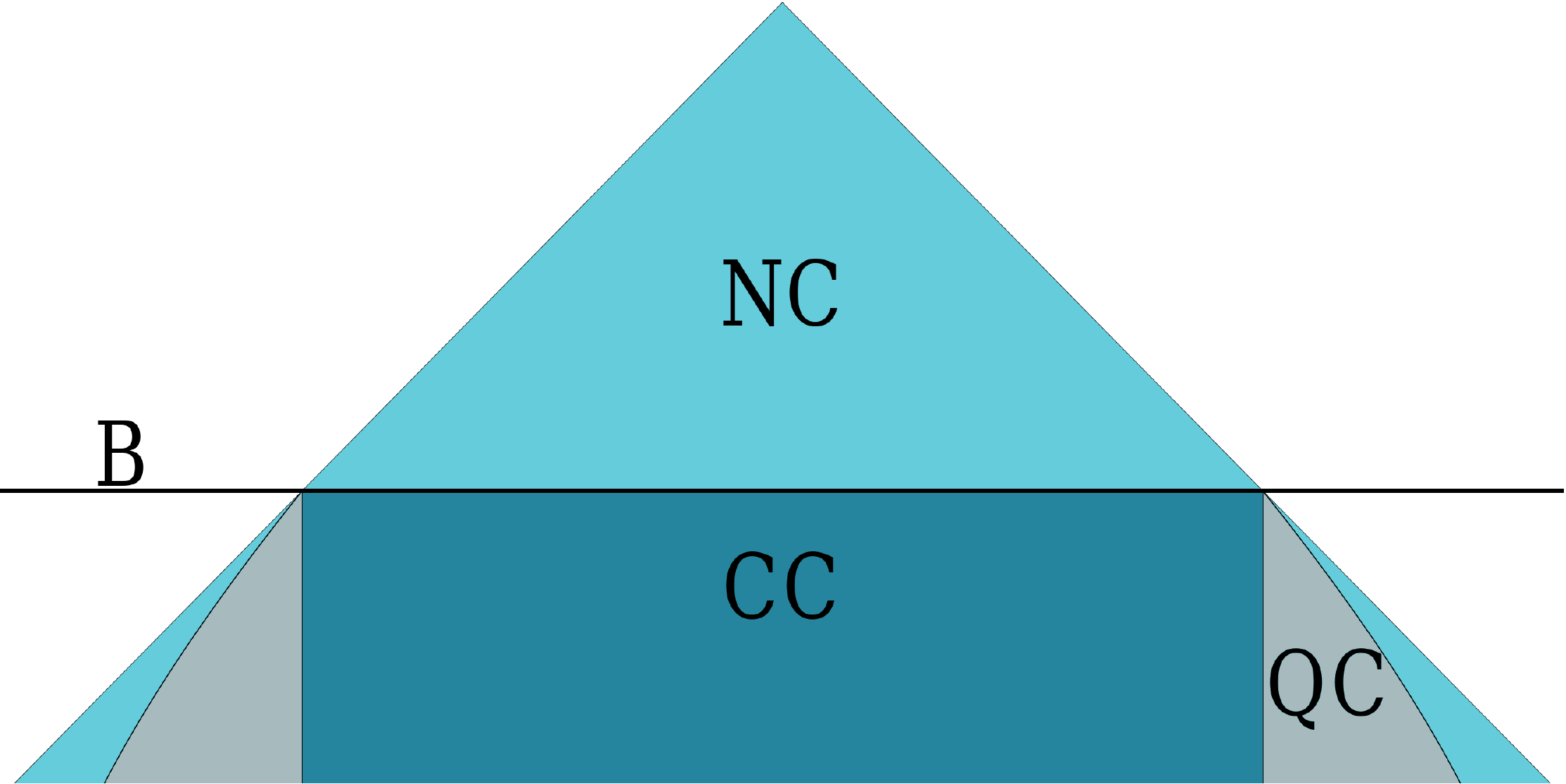}
\caption{Schematic depiction of the sets of classical $CC$,
quantum $QC$ and non-signalling correlations $NC$. Tight Bell
inequalities correspond to facets of the classical set. $B$
denotes a Bell inequality with no quantum violation which is a)
not tight, b) tight. Note that only tightness guarantees the
existence of a non-trivial region in which quantum and classical
correlations coincide.}\label{fig:tight}
\end{figure}
Recall that a Bell inequality is called \textit{tight} when it
defines a facet of the convex set of CC (see e.g.~Ref.~\cite{polytopes}).

To our knowledge, the first nontrivial tight Bell inequalities
which are not violated by quantum theory were those proposed in
Ref.~\cite{Mafalda}. From a geometric point of view, the existence
of such Bell inequalities implies that the convex sets of quantum
and classical correlations can share facets. These inequalities
were also used in a different context to prove that, contrary to
the bipartite scenario~\cite{Barnum, Hadley}, local quantum
measurements and the no-signalling principle do not imply that
correlations are quantum in a general multipartite
scenario~\cite{Hadley}.

More recently, some of us have proposed a systematic construction
of nontrivial Bell inequalities with no quantum
violation~\cite{my}. The construction exploits the concept of
unextendible product bases (UPBs)~\cite{BennettUPB}. This
connection is remarkable, as UPBs are a notion of entanglement
theory and heavily rely on the structure of tensor products of
Hilbert spaces. Interestingly, this construction reproduces the
Bell inequalities previously derived in~\cite{Mafalda}, thus
proving that it may lead to tight Bell inequalities with no
quantum violation.  Unfortunately, these have so far been the only
examples of tight Bell inequalities found via the construction.
The main aim of this paper is to provide new examples of tight
Bell inequalities with no quantum violation that arise from UPBs.
To this end, we discuss in more detail how Bell inequalities can
be derived from any set of mutually orthogonal product vectors in
many-qubit Hilbert spaces and prove that the concept of
unextendibility plays a crucial role for the nontriviality of the
associated Bell inequality. In particular, we show that the only
nontrivial Bell inequalities that can be constructed in this way
are those coming from UPBs or sets that can be completed only to a
UPB. We also prove that the Bell inequalities from
Ref.~\cite{Mafalda} are tight for any odd number of parties.

This paper is structured as follows. Sec.~\ref{Preliminaries}
introduces all concepts relevant for the upcoming sections. In
Sec.~\ref{Construction}, we recall the construction from~\cite{my}
which associates to a set of mutually orthogonal product vectors a
Bell inequality and study its properties in more detail. We prove
that whenever the initial set of product vectors can be completed
to a full basis of product vectors, then the resulting Bell
inequality is trivial in the sense that it cannot be violated by
any NC. In Sec.~\ref{GYNI} we show in more detail that the Bell
inequalities from~\cite{Mafalda} can be constructed from UPBs and
prove their tightness for any odd $n$.  Sec.~\ref{New} presents
new tight Bell inequalities with no quantum violation constructed
from UPB, while Sec.~\ref{Conclusion} concludes the paper.

\section{Preliminaries}
\label{Preliminaries}

Before getting to the results, let us first establish some
terminology and notation and recall concepts and facts concerning
unextendible product bases and nonsignalling correlations.

By abuse of terminology, we use the term ``vector'' in the context
of quantum states always in the sense of ``$1$-dimensional
subspace of a complex Hilbert space''. Phrased differently, this
means that we take our vectors to be unit vectors, and we identify
two unit vectors whenever they differ only by a complex phase. By
``basis'', we always mean an orthonormal basis.

\subsection{Unextendible product bases}
\label{UPB}

Consider a product $n$-partite Hilbert space
\begin{equation}\label{HilbertSpace}
\mathcal{H}=\mathbbm{C}^{d_1}\ot\ldots\ot\mathbbm{C}^{d_n}
\end{equation}
with $d_i$ $(i=1,\ldots,n)$ denoting local dimensions. Following
Ref.~\cite{BennettUPB}, an \textit{unextendible product basis}
(UPB) is a collection of mutually orthogonal fully product vectors
in $\mathcal{H}$,
\begin{equation}
U=\left\{\ket{\phi^{(1)}_j}\ot\ldots\ot\ket{\phi^{(n)}_j}\right\}_{j=1}^{|U|},
\end{equation}
obeying two conditions: (i)
$|U|<\mathrm{dim}\mathcal{H}=\Pi_{i=1}^{n}d_i$ ($U$ does not span
$\mathcal{H}$), and (ii) $(\mathrm{span}\,U)^{\perp}$ does not
contain any product vector, or, in other words, is a completely
entangled subspace.

UPBs were introduced in the context of entanglement
theory in Ref.~\cite{BennettUPB}, where they were used to obtain
one of the first constructions of bound entangled states
\cite{BennettUPB,UPBhuge}. More precisely, the state
\begin{equation}\label{upbstate}
    \rho_{U}=\frac{\mathbbm{1}_{\mathcal{H}}-\Pi_{U}}{D-|U|}\,,
\end{equation}
where $\Pi_{U}$ denotes the sum of projectors onto vectors from
$U$, has positive partial transpose with respect to any
bipartition, but nevertheless is entangled. While the former
follows from the fact that the application of partial
transposition with respect to any subset of parties to $\Pi_{U}$
returns another projector, the latter is a consequence of the lack
of product vectors in $(\mathrm{span}\,U)^{\perp}$ and hence the
range criterion applies here~\cite{PHrange}. UPBs are interesting
and intriguing objects and hence there has been some effort
towards understanding their properties and structure (see e.g.
Refs~\cite{BennettUPB,UPBhuge,upb, Bravyi}).

To illustrate the above definition, let us provide two examples of
UPBs.

\begin{example}
\label{pyramid} First, we consider one of the earliest examples of
a bipartite UPB, the so-called \textit{pyramid}~\cite{BennettUPB},
$\mathcal{U}_{\mathrm{Pyr}}=\left\{\ket{v_{j}}\ot\ket{v_{2j\,\mathrm{mod}\,
5}}\right\}_{j=0}^{4}\subset \mathbbm{C}^3\ot\mathbbm{C}^3$ with
\begin{equation}\label{vj}
\ket{v_j}=N\left(\cos\varphi_j\ket{0}+\sin\varphi_j\ket{1}+h\ket{2}\right),
\end{equation}
where $N=2/\sqrt{5+\sqrt{5}}$, $h=(1/2)\sqrt{1+\sqrt{5}}$, and
$\varphi_j=2\pi j/5$. One easily finds that there is no product
vector in $(\mathrm{span}\,\mathcal{U}_{\mathrm{Pyr}})^{\perp}$
and hence $\mathcal{U}_{\mathrm{Pyr}}$ is a UPB in
$\mathcal{H}=\mathbbm{C}^{3}\ot\mathbbm{C}^3$. Let us also note
that in the bipartite case, $\mathcal{U}_\mathrm{Pyr}$ is the
lowest-dimensional example of a UPB: in
$\mathbbm{C}^2\ot\mathbbm{C}^d$ there are no UPBs for any $d$.
\end{example}

\begin{example} Let us now take the following
four-element set of three-qubit vectors \cite{BennettUPB}:
\begin{equation}\label{Schifts}
\mathcal{U}_{\mathrm{Shifts}}=\{\ket{000},\ket{1\overline{e}e},\ket{e1\overline{e}},\ket{\overline{e}e1}\}
\end{equation}
where $\ket{e}\in\mathbbm{C}^2$ is an arbitrary unit vector
different from $\ket{0}$ and $\ket{1}$ and $\ket{\overline{e}}$
stands for the unit vector orthogonal to $\ket{e}$ (unique up to
phase). This set is a slight generalization of the so-called
``Shifts'' UPB found in Ref.~\cite{BennettUPB} and then
generalized to more parties in Ref.~\cite{UPBhuge}. Notice also
that for each qubit, we can replace
$\{\ket{e},\ket{\overline{e}}\}$ with a different basis
$\{\ket{e_i},\ket{\overline{e}_i}\}$ $(i=1,2,3)$ independent of
$\ket{0}$ and $\ket{1}$. Up to local unitary equivalence, there
are no other UPBs in
$\mathbbm{C}^2\ot\mathbbm{C}^2\ot\mathbbm{C}^2$~\cite{Bravyi}.
\end{example}

\subsection{Nonsignalling correlations}
\label{NC}

Let us consider $n$ observers having access to $n$ correlated
systems. The $i$th observer $(i=1,\ldots,n)$ can perform on his
system one of $m_i$ possible measurements with $r_i^{x_i}$
outcomes, henceforward denoted $a_i$, where
$x_i\in\{0,\ldots,m_i-1\}$ stands for the measurement choice of
the $i$th observer. The correlations established in this way are
determined by the collection of conditional probabilities
\begin{equation}\label{CP}
\{p(\boldsymbol{a}|\boldsymbol{x})\equiv p(a_1,\ldots,
a_n|x_1,\ldots, x_n)\},
\end{equation}
where $\boldsymbol{a}=(a_1,\ldots, a_n)$ and
$\boldsymbol{x}=(x_1,\ldots,x_n)$. The usual way of dealing with
these objects is to treat them as a vector in $\mathbbm{R}^D$ with
$D=\prod_{i=1}^{n}\sum_{x_i=0}^{m_i-1}r_i^{x_i}.$

Clearly, the probabilities $p(\boldsymbol{a}|\boldsymbol{x})$ are
nonnegative and normalized in the sense that
$\sum_{\boldsymbol{a}}p(\boldsymbol{a}|\boldsymbol{x})=1$ holds
for any $\boldsymbol{x}$. Additionally, as it is assumed that no
communication among the parties can take place when the
measurements are performed, the obtained correlations must obey
the principle of no-signalling: the choices of observables by a
set of parties cannot influence the statistics seen by the
remaining parties. Formally, this can be stated as a set of
equations of the form
\begin{equation}
\label{nonsignalling}
\begin{split}
&\sum_{a_i}p(a_1\ldots a_i\ldots a_n|x_1\ldots x_i\ldots
x_n)\\
=&\sum_{a_i}p(a_1\ldots a_i\ldots a_n|x_1\ldots x'_i\ldots x_n)
\end{split}
\end{equation}
for all $x_i,x'_i$, and $a_1,\ldots,a_{i-1},a_{i+1},\ldots,a_n$
and $x_1,\ldots,x_{i-1},x_{i+1},\ldots,x_n$ and all $i$. The
conditional probabilities~(\ref{CP}) constrained by the
positivity, normalization and the nonsignalling
conditions~(\ref{nonsignalling}) form a polytope (see
e.g.~Ref.~\cite{polytopes}) whose dimension depends on the
considered scenario and is given by~\cite{Lifting}:
\begin{equation}
d=\prod_{i=1}^{n}\left[\sum_{x_i=0}^{m_i-1}(r_i^{x_i}-1)+1\right]-1.
\end{equation}

\paragraph*{Quantum correlations (QC).} Assume now that the parties have
access to correlated quantum particles. The resulting
correlations are then guaranteed to satisfy the no-signalling equations and read
\begin{equation}\label{QC}
p(a_1,\ldots, a_n|x_1,\ldots, x_n)=\Tr\left[\varrho\,
P_{x_1}^{a_1}\ot\ldots\ot P_{x_n}^{a_n}\right],
\end{equation}
where $\varrho$ stands for a density matrix and $P_{x_i}^{a_i}$
denote positive operators representing the measurement outcomes at
the $i$th site. For each $x_i$, they need to satisfy
\begin{equation}\label{Conditions}
\sum_{a_i}P_{x_i}^{a_i}=\mathbbm{1}\qquad (i=1,\ldots,n) .
\end{equation}
Since any quantum measurement can be realized as a
projective measurement on a Hilbert space of sufficiently large
dimension, we can always assume that all
$P_{x_i}^{a_i}$ are orthogonal projectors.

\paragraph*{Classical correlations.} Let us now consider
correlations that can be established by the $n$ observers when
they have access only to shared classical information in the form
of shared randomness $\lambda$, which is a random variable with
arbitrary distribution $p(\lambda)$. This defines the set of {\it
classical correlations} (CC). It is the set of all those
conditional probabilities which can be written in the form
\begin{equation}\label{local}
p(\boldsymbol{a}|\boldsymbol{x})=\sum_{\lambda}p(\lambda)\prod_{i=1}^{n}p_i(a_i|x_i,\lambda).
\end{equation}
where each $p_i(a_i|x_i,\lambda)$ is an arbitrary conditional
probability distribution. It follows that, analogously to the case
of NC, the set of all CC is a polytope in the same space. Its
extremal points are the deterministic probabilities
$p(\boldsymbol{a}|\boldsymbol{x})=\prod_{i=1}^{n}p_i(a_i|x_i)$
where each probability $p_i(a_i|x_i)$ equals either zero or one.

The set of CC is strictly smaller than the set of QC
which~\cite{Bell}, in turn, is strictly smaller than the set of
NC~\cite{PR}.

\paragraph*{Bell inequalities.} Consider a linear combination of the
conditional probabilities,
$\sum_{\boldsymbol{a},\boldsymbol{x}}T_{\boldsymbol{a},\boldsymbol{x}}p(\boldsymbol{a}|\boldsymbol{x})$
with $T_{\boldsymbol{a},\boldsymbol{x}}$ being some $2n$-index
tensor $T_{\boldsymbol{a},\boldsymbol{x}}$. Finding the maximal
value, written as $\beta_{C}$, of this expression over all local
probabilities (\ref{local}), one arrives at {\it the Bell
inequality}~\cite{Bell}:
\begin{equation}\label{generalBellIneq}
\sum_{\boldsymbol{a},\boldsymbol{x}}T_{\boldsymbol{a},\boldsymbol{x}}p(\boldsymbol{a}|\boldsymbol{x})\leq
\beta_{C}.
\end{equation}
We say that a Bell inequality is {\it nontrivial} if it is
violated by some NC, that is there exist NC,
$p(\boldsymbol{a}|\boldsymbol{x})$, such that
$\sum_{\boldsymbol{a},\boldsymbol{x}}T_{\boldsymbol{a},\boldsymbol{x}}p(\boldsymbol{a}|\boldsymbol{x})>\beta_C$.

Geometrically, nontrivial Bell inequalities are hyperplanes that
separate CC from some NC, and possibly also from some QC. A Bell
inequality is said to be {\it tight} whenever it defines a facet
of the polytope of CC. Like any other polytope~\cite{polytopes},
the polytope of CC can be fully described in terms of its facets,
that is by all the tight Bell inequalities. If some correlations
do not have a classical realization, they necessarily violate a
tight Bell inequality. This explains our interest in tightness.

Given a Bell inequality, how does one find out whether it is tight
or not? It is tight if and only if those CC which saturate the
Bell inequality span, when treated as vectors from
$\mathbbm{R}^D$, an affine subspace of dimension $d-1$. So in
order to check tightness, one has to see whether the models
attaining the maximum value $\beta_C$ constitute a set of $d$
linearly independent vectors (for more detail see e.g.
Refs.~\cite{Lifting,Lluis}).

Finally, let us mention that when rewritten in an appropriate form
(entries of $T_{\boldsymbol{a},\boldsymbol{x}}$ are nonnegative
and normalized), every Bell inequality can be understood as a
nonlocal game as follows. Upon receiving, in a distributed manner,
the input $\boldsymbol{x}$ (from some fixed set of possible
inputs), the parties determine an output $\boldsymbol{a}$ and
receive a payoff $T_{\boldsymbol{a},\boldsymbol{x}}$. Then, the
left-hand side of (\ref{generalBellIneq}) corresponds to the value
of the game. Accordingly, the classical bound $\beta_C$ stands for
the maximal value of the game in the case when the only resource
at a disposal of the parties is a shared randomness. Then,
violation of a Bell inequality by some QC means that there exist
quantum resources allowing the parties to perform the
corresponding task with greater efficiency than allowed by
classical physics.

\section{Bell inequalities with no quantum violation from UPBs}
\label{Construction}

In this section we recall and study in some more detail the scheme
from Ref.~\cite{my} for constructing Bell inequalities from sets
of orthogonal product vectors. Also, we prove that if the set of
orthogonal product vectors can be completed to a full basis of
product vectors (or already constitutes a full basis itself), then
the associated Bell inequality is trivial in the sense that it
cannot be violated by any NC.

\subsection{The construction}
\label{constr}

We now restrict to an $n$-qubit Hilbert space
$\mathcal{H}=(\C^2)^{\otimes n}$. Suppose that
\begin{equation}
\mathcal{S}=\left\{\ket{\psi_j}=\ket{\psi_j^{(1)}}\ot\ldots\ot\ket{\psi_j^{(n)}}\right\}_{j=1}^{|\mathcal{S}|},
\end{equation}
is a set of product vectors in $\mathcal{H}$, where each
$\ket{\psi_j^{(i)}}\in\C^2$ is a unit vector. In the following, we
assume that the $\ket{\psi_j}$ are mutually orthogonal. This
implies an upper bound on the number of elements of $\mathcal{S}$
given by $|\mathcal{S}|\leq 2^n$. For the time being, however, we
do not assume $\mathcal{S}$ to be a UPB.

For each $i$, we now take the $\ket{\psi^{(i)}_j}$ to be ordered
in such a way that the vectors in the local set
$$
\mathcal{S}^{(i)} = \left\{ \ket{\psi^{(i)}_1},\ldots,\ket{\psi^{(i)}_{s_i}} \right\}
$$
are all different, and such that each $\ket{\psi^{(i)}_j}$ for
$j>s_i$ is already contained in this list. In general, either  of
$s_i<|\mathcal{S}|$ or $s_i=|\mathcal{S}|$ is possible.

Then we partition each $\mathcal{S}^{(i)}$ into disjoint subsets
$$
\mathcal{S}^{(i)}=\mathcal{S}^{(i)}_1\cup \ldots\cup\mathcal{S}^{(i)}_{m_i}
$$
such that two vectors in $\mathcal{S}^{(i)}$ are orthogonal if and
only if they lie in the same subset of the partition. This is
possible because of the following property of $\C^2$: if
$\ket{\phi'}$ is neither orthogonal to nor equal to $\ket{\phi}$,
then it is also neither orthogonal to nor equal to
$\ket{\phi^\perp}$. Alternatively speaking, orthogonality is an
equivalence relation on vectors in $\C^2$.

In the framework of~\cite{my}, this has been called
\textit{property (P)}, and it has also been noted that it is
automatic in the qubit case. Since here we consider the qubit case
only, each local subset $\mathcal{S}^{(i)}_j$ contains at most $2$
vectors. As example~\ref{pyramid}  shows, there exist sets of
orthogonal product vectors without property (P) when the local
Hilbert spaces have dimension higher than $2$.

Resuming the construction of the Bell inequality, we also need to
fix an arbitrary ordering of each subset $\mathcal{S}^{(i)}_j$.

Now the Bell scenario associated to this $\mathcal{S}$ is given
by $n$ parties, where party $i$ has $m_i$ measurement settings and
every measurement has two possible outcomes.

The Bell inequality associated to $\mathcal{S}$ in this scenario
is defined as follows. We assign to every product vector
$\ket{\psi_j}\in\mathcal{S}$ a certain term
$p(\boldsymbol{a}_j|\boldsymbol{x}_j)$. This term is defined by a
list of settings, $\boldsymbol{x}_j=(x_j^{(1)},\ldots,x_j^{(n)})$,
and a list of outcomes,
$\boldsymbol{a}_j=(a_j^{(1)},\ldots,a_j^{(n)})$. These are
obtained through the procedure of

\begin{itemize}
    \item determining the local vector $\ket{\psi_j^{(i)}}$ appearing at the $i$th site of $\ket{\psi_j}$ for each $i$,

    \item setting $x_j^{(i)}$ to be the number $k$ for which $\ket{\psi_j^{(i)}}\in\mathcal{S}_{k}^{(i)}$,
while taking $a^{(i)}_j$ to be the position of
$\ket{\psi_{j}^{(i)}}$ within $\mathcal{S}_{k}^{(i)}$.
\end{itemize}

Let us now take a linear combination of these terms, $\beta=\sum_i
q_i p(\boldsymbol{a}_i|\boldsymbol{x}_i)$, with nonnegative
weights $q_i$, which always can be assumed to obey $0\leq q_i\leq
1$. This leads us to the Bell inequality
\begin{equation}\label{BellIneq}
\sum_{i=1}^{|\mathcal{S}|}q_i
p(\boldsymbol{a}_i|\boldsymbol{x}_i)\leq \beta_C,
\end{equation}
where, as before, $\beta_C$ stands for the maximal value of the
left-hand side of~(\ref{BellIneq}) over classical probability
distributions~(\ref{CP}).

Our main aim throughout the present paper is to investigate these
inequalities, and, in particular, how their properties are related
to the properties of the underlying sets $\mathcal{S}$. First, we
prove that any such inequality cannot be violated by quantum
theory. Denoting by $\beta_Q$ the maximal value achievable by the
left-hand side of~(\ref{BellIneq}) within quantum theory,
we have the following fact.\\

\begin{thm}\label{theorem1}
Let $\mathcal{S}$ be the set of
mutually orthogonal vectors from $\mathcal{H}$.
Then for the corresponding Bell inequality (\ref{BellIneq}),
$\beta_C=\beta_Q=\max\{q_i\}$.
\end{thm}

\begin{proof}
Orthogonality of any pair of vectors from $\mathcal{S}$ implies
that at some position they have different vectors from the same
local subset $\mathcal{S}^{(i)}_k$. This means that any pair of
the associated conditional probabilities has at some site the same
inputs but different outputs (different outcomes of the the same
observable). Consequently, for any deterministic local model
(recall that to get $\beta_C$ one can restrict to these models) if
one of the conditional probabilities equals unity, the remaining
ones vanish (let us call such probabilities ``orthogonal'').
Consequently $\beta_C=\max\{q_i\}$.

In order to prove that $\beta_Q=\max\{q_i\}$, we can always
restrict to local von Neumann measurements and assign the product
projectors
\begin{equation}
P_j=\bigotimes_{i=1}^{n}P_{j}^{(i)}
\end{equation}
to the conditional probabilities
$p(\boldsymbol{a}_j|\boldsymbol{x}_j)$. Then, ``orthogonality'' of
the conditional probabilities directly implies that $P_j\perp P_k$
for $j\neq k$. Consequently, the corresponding Bell operator
\begin{equation}\label{BellOp}
B=\sum_{j=1}^{|\mathcal{S}|}q_j\bigotimes_{i=1}^{n} P_{j}^{(i)}.
\end{equation}
is a positive operator whose eigenvalues are the $q_j$'s, meaning that $\beta_Q=\max\{q_j\}$.
This completes the proof.
\end{proof}

\begin{cor}
An immediate consequence of this fact is that for a given set
$\mathcal{S}$, the strongest Bell inequality it
generates~(\ref{BellIneq}) is the one with equal $q$'s. This is
because the left-hand side of~(\ref{BellIneq}) can always be upper
bounded by
$\max\{q_i\}\sum_{j=1}^{|\mathcal{S}|}p(\boldsymbol{a}_j|\boldsymbol{x}_j)$
and the latter expression give rise to a Bell inequality with the
same classical bound as~(\ref{BellIneq}). Along the same lines, it
is fairly easy to see that a Bell inequality with unequal $q$'s
cannot be tight. Consequently, from now on we restrict to Bell
inequalities with equal $q$'s, i.e.~to
\begin{equation}\label{BellIneqEqual}
\sum_{i=1}^{|\mathcal{S}|}
p(\boldsymbol{a}_i|\boldsymbol{x}_i)\leq 1.
\end{equation}
\end{cor}

Let us illustrate the above construction by applying it to
two particular sets of vectors.

\begin{example}
First, consider the set
$\mathcal{U}_{\mathrm{Shifts}}$~(\ref{Schifts}). Clearly, it has
the property (P) and at each site we can distinguish two local
sets $\mathcal{S}_0^{(i)}\equiv
\mathcal{S}_{0}=\{\ket{0},\ket{1}\}$ and
$\mathcal{S}_1^{(i)}\equiv
\mathcal{S}_{1}=\{\ket{e},\ket{\overline{e}}\}$ $(i=1,2,3)$. Then,
to each element of $\mathcal{U}_{\mathrm{Shifts}}$ we assign
conditional probabilities in the following way: $\ket{000}\mapsto
p(000|000)$, $\ket{1\overline{e}e}\mapsto p(110|011)$,
$\ket{e1\overline{e}}\mapsto p(011|101)$, and
$\ket{\overline{e}e1}\mapsto p(101|110)$. Summing up these
probabilities we get the tight Bell inequality with no quantum
violation found in Ref.~\cite{Sliwa} and studied in
Ref.~\cite{Mafalda}:
\begin{equation}\label{Ex1}
p(000|000)+p(110|011)+p(101|110)+p(011|101)\leq 1.
\end{equation}
\end{example}

\begin{example}
Second, we take the full basis in the three-qubit Hilbert space
$\mathcal{H}=(\mathbbm{C}^2)^{\ot 3}$ giving rise to the
phenomenon called nonlocality without
entanglement~\cite{BennettNNWE}:
\begin{equation}
\mathcal{S}=\{\ket{000},\ket{e01},\ket{01e},
\ket{01\overline{e}},\ket{1e0},\ket{\overline{e}01},
\ket{1\overline{e}0},\ket{111}\}.
\end{equation}
As before, we partition the set of local vectors into subsets; in
this case, the local vectors are the same for each party, and the
subsets are
$$
\mathcal{S}_0=\left\{|0\rangle,|1\rangle\right\} ,\quad  \mathcal{S}_1=\left\{|e\rangle,|\overline{e}\rangle\right\}.
$$
Then, associating conditional probabilities to every element of
$\mathcal{S}$ and following the above procedure, we arrive at
the Bell inequality
\begin{eqnarray}\label{BI2}
&&p(000|000)+p(001|100)+p(010|001)+p(011|001)\nonumber\\
&&+p(100|010)+p(101|100)+p(110|000)\nonumber\\
&&+p(111|000)\leq 1.
\end{eqnarray}
This one, however, is trivial as it cannot be violated by any NC.
Indeed, using nonsignalling conditions (\ref{nonsignalling}), it
can be shown that the left-hand side of Eq.~(\ref{BI2}) is exactly
one. This, as we will see shortly, is a consequence of the fact
that $\mathcal{S}$ is a basis in $(\mathbbm{C}^{2})^{\ot 3}$.
\end{example}

These two examples reflect the importance of the notion of UPB in
our construction. Any set of product orthogonal vectors
$\mathcal{S}\subseteq \mathcal{H}$ is either a UPB (or a set that
can be completed only to a UPB), or completable to a full basis in
$\mathcal{H}$ (or already one). In the second case, we will call
the set {\it completable}. Interestingly, any Bell inequality
constructed from the completable set is trivial in the sense that
it cannot be violated by any NC. On the other hand, Bell
inequalities associated to UPBs are always nontrivial as there
exist NC violating them. The following two theorems formalize the above statements. \\

\begin{thm}\label{theorem2}
Let $\mathcal{S}$ be a set of orthogonal product vectors in an
$n$-qubit Hilbert space $\mathcal{H}$. If $\mathcal{S}$ is
completable, then the resulting Bell inequality~(\ref{BellIneq})
cannot be violated by any NC.
\end{thm}

\begin{proof}
Let us start by noting that any Bell inequality which is saturated
by some interior point of the CC polytope is trivial. Then, the
left-hand side of~(\ref{generalBellIneq}) has a constant value,
equal to $\beta_C$, on the whole affine subspace spanned by all
CC, and this subspace contains all nonsignalling correlations.

To see this more explicitly, assume that a given Bell
inequality~(\ref{generalBellIneq}) is saturated by an interior
point $\{\widetilde{p}(\boldsymbol{a}|\boldsymbol{x})\}$ of the
corresponding CC polytope, i.e.,
\begin{equation}\label{Satur}
\sum_{\boldsymbol{a},\boldsymbol{x}}T_{\boldsymbol{a},
\boldsymbol{x}}\widetilde{p}(\boldsymbol{a}|\boldsymbol{x})=
\beta_{C}.
\end{equation}
We can always represent
$\{\widetilde{p}(\boldsymbol{a}|\boldsymbol{x})\}$ as a convex
combination of some extremal point (deterministic local point),
denoted $\{p_{\mathrm{ex}}(\boldsymbol{a}|\boldsymbol{x})\}$, of
this polytope and some other point lying on its boundary. This,
when substituted into Eq.~(\ref{Satur}), directly implies that
$\{p_{\mathrm{ex}}(\boldsymbol{a}|\boldsymbol{x})\}$ must
saturate~(\ref{generalBellIneq}). Since any extremal point of the
CC polytope can be used in this
decomposition,~(\ref{generalBellIneq}) must be saturated by all of
them. Consequently, any affine combination of the vertices of the
CC polytope, and in particular any NC, also
saturates~(\ref{generalBellIneq}).

Let us now assume that $\mathcal{S}$ is a full basis in an $n$-qubit
Hilbert space $\mathcal{H}$ and consider the associated Bell
inequality
\begin{equation}\label{AnotherBellIneq}
\sum_{j=1}^{\mathrm{dim}\mathcal{H}}p(\boldsymbol{a}_j|\boldsymbol{x}_j)\leq
1.
\end{equation}
Then, let us take the uniform probability distribution, i.e.,
$p(\boldsymbol{a}|\boldsymbol{x})=1/\mathrm{dim}\mathcal{H}$ for
any $\boldsymbol{a}$ and $\boldsymbol{x}$. On one hand, it clearly
saturates the Bell inequality~(\ref{AnotherBellIneq}). On the
other hand, it belongs to the interior of the corresponding CC
polytope, which, in view of what we have just said, implies
that~(\ref{AnotherBellIneq}) cannot be violated by any NC.

Finally, let us assume that $\mathcal{S}$ is not a full basis in
$\mathcal{H}$, but can be completed to one. Let us write
$\overline{\mathcal{S}}=\{\ket{\phi_j}\}_{j=1}^{\mathrm{dim}\mathcal{H}-|\mathcal{S}|}$
for the completing set of mutually orthogonal product vectors.
Again $\mathcal{S}\cup \overline{\mathcal{S}}$ is a set of
mutually orthogonal product vectors, and therefore has an
associated Bell inequality
$$
\sum_{j=1}^{\mathrm{dim}\mathcal{H}} p(\boldsymbol{a}_j|\boldsymbol{x}_j) \leq 1 .
$$
In the previous paragraph, we showed that this inequality is
trivial in the sense that all NC satisfy it with equality. Now we
can upper bound the left-hand side of~(\ref{BellIneqEqual}) for
any NC $p(\boldsymbol{a}|\boldsymbol{x})$ by
\begin{equation}
\sum_{j=1}^{|\mathcal{S}|}p(\boldsymbol{a}_j|\boldsymbol{x}_j)\leq
\sum_{j=1}^{\mathrm{dim}\mathcal{H}}p(\boldsymbol{a}_j|\boldsymbol{x}_j) = 1.
\end{equation}
Therefore, any NC satisfy~(\ref{BellIneq}). This completes the proof.
\end{proof}

\begin{rem}
For completeness, let us present here an alternative but less
general way of proving the above statement. It uses the particular
form of the considered Bell inequalities~(\ref{BellIneqEqual}) and
provides some additional insight into the relationship between
UPBs and Bell inequalities.

Let us assume that $\mathcal{S}$ is a full basis in $\mathcal{H}$
but the associated Bell inequality~(\ref{BellIneqEqual}) is not
saturated by some extreme point
$\{\widetilde{p}(\boldsymbol{a}|\boldsymbol{x})\}$ of the
associated polytope of CC. This means that for this point (recall
that probabilities representing deterministic local models can
only equal zero or one),
\begin{equation}
\label{trivialineq}
\sum_{j=1}^{\mathrm{dim}\mathcal{H}}\widetilde{p}(\boldsymbol{a}_j|\boldsymbol{x}_j)=0,
\end{equation}
and consequently
\begin{equation}\label{12345}
\widetilde{p}(\boldsymbol{a}_j|\boldsymbol{x}_j)=0
\end{equation}
for every $j=1,\ldots,\mathrm{dim}\mathcal{H}$. By assumption,
$\{\widetilde{p}(\boldsymbol{a}|\boldsymbol{x})\}$ is local and
deterministic, so that
$\widetilde{p}(\boldsymbol{a}|\boldsymbol{x})=\widetilde{p}(a_1|x_1)\ldots
\widetilde{p}(a_n|x_n)$ for any $\boldsymbol{a}$ and
$\boldsymbol{x}$.

For concreteness, but without loss of generality, we can assume
that whenever $\widetilde{p}(a_i|x_i)=0$ for some observable $x_i$
at site $i$, then $a_i=0$. This amounts to labelling the outcomes
of each observable such that $\widetilde{p}(1|x_i)=1$.

Now the orthogonality of the vectors in $\mathcal{S}$, from which
the inequality was constructed, means that any two outcome strings
$\boldsymbol{a}_{j_1}$ and $\boldsymbol{a}_{j_2}$ which appear
in~(\ref{trivialineq}) are different: any two vectors in
$\mathcal{S}$ have orthogonal components at some site $i$; at this
site $i$, the associated outcome strings necessarily have to be
different, so that $a_{j_1,i}\neq a_{j_2,i}$ (where now the second
index enumerates the parties, $i=1,\ldots,n$). However, since
there are only $2^n=\mathrm{dim}\mathcal{H}$ possible outcome
strings, all of these do appear in~(\ref{trivialineq}). In
particular, also the constant outcome string $1\ldots 1$ occurs
in~(\ref{12345}) for some $j$. This is a contradiction to
$\widetilde{p}(1\ldots 1|\boldsymbol{x})=1$ for all
$\boldsymbol{x}$.

Therefore, the Bell inequality~(\ref{BellIneqEqual}) associated to
a full basis in $\mathcal{H}$ is saturated by all the extreme
points of the corresponding CC polytope and hence trivial.
\end{rem}

\begin{thm}\label{theorem3}
Let $\mathcal{S}$ be an $n$-qubit UPB. Then the corresponding Bell
inequality~(\ref{BellIneqEqual}) is nontrivial, i.e., there exist
NC violating it.
\end{thm}

\begin{proof}
Let us denote by $\Pi_{\mathrm{UPB}}$ the projector onto the UPB
$\mathcal{S}$, and introduce a normalized entanglement witness
\begin{equation}\label{witness}
W=\frac{1}{|\mathcal{S}|-\mathrm{dim}\mathcal{H}}\left(\Pi_{\mathrm{UPB}}-\epsilon
\mathbbm{1}_{\mathcal{H}}\right),
\end{equation}
where $\epsilon$ is the positive number defined as
\begin{equation}
\epsilon=\min_{\ket{\psi_{\mathrm{prod}}}\in\mathcal{H}}
\langle\psi_{\mathrm{prod}}|\Pi_{\mathrm{UPB}}|\psi_{\mathrm{prod}}\rangle
\end{equation}
with the minimum going over all fully product vectors in
$\mathcal{H}$ and $\mathbbm{1}_{\mathcal{H}}$ standing for the
identity acting on $\mathcal{H}$.

This witness detects entanglement of the bound entangled
state~(\ref{upbstate}) in the sense that
\begin{equation}
\Tr\left(W\rho_{\mathcal{S}}\right)<0.
\end{equation}
Using the explicit form of $\rho_{\mathcal{S}}$ [cf.
Eq.~(\ref{upbstate})], we can rewrite the above in the following
way
\begin{equation}
\Tr\left(W\Pi_{\mathrm{UPB}}\right)>1.
\end{equation}
To complete the proof, one first notices that $\Pi_{\mathrm{UPB}}$
is exactly a Bell operator corresponding to Bell
inequality~(\ref{BellIneqEqual}) constructed from the UPB
$\mathcal{S}$. Second, it is known that any entanglement witness
represents some nonsignalling correlations (see e.g.
Ref.~\cite{Hadley}). These two facts together imply that the Bell
inequality constructed from the UPB $\mathcal{S}$ is violated by
some NC.
\end{proof}

\begin{cor}
An immediate consequence of both the above theorems is that a Bell
inequality~(\ref{BellIneqEqual}) constructed from $\mathcal{S}$ is
nontrivial iff $\mathcal{S}$ is a UPB or can be completed only to
a UPB (in the sense that there still exists a product vector in
$\mathcal{H}$ orthogonal to $\mathcal{S}$, however, one is unable
to complete $\mathcal{S}$ to a full basis in $\mathcal{H}$).
\end{cor}

It should also be noted that one can easily increase the violation
of~(\ref{BellIneq}) by the $W$ from~(\ref{witness}) with
$\epsilon$ replaced by
\begin{equation}\label{espilon2}
\epsilon'=\min_{\ket{\psi_{\mathrm{prod}}}\in \mathcal{S}'}\langle
\psi_{\mathrm{prod}}|\Pi_{\mathcal{S}}|\psi_{\mathrm{prod}}
\rangle,
\end{equation}
where $\mathcal{S}'$ is a finite set of product vectors
constructed by taking all the possible tensor products of vectors
from the local sets $\mathcal{S}^{(i)}$. In other words,
$\mathcal{S}'$ is a set of product vectors representing local
measurement operators and allowing one to obtain NC from $W$. The
nonnegativity condition of these conditional probabilities leads
to~(\ref{espilon2}). For instance, for the Bell
inequality~(\ref{Ex1}) one can put $\epsilon'=1/8$ and get a
violation $7/6$ which is still, however, less than the maximal
nonsignalling violation of $4/3$~\cite{Mafalda}.

\section{Guess your neighbour's input Bell inequalities and
$n$-qubit UPBs}\label{GYNI}

Our construction connects UPBs and nontrivial Bell inequalities.
It is now tempting to ask whether it is possible to get {\it
tight} Bell inequalities in this way. In Ref.~\cite{my}, some of
us showed that the recently introduced tight Bell inequalities
\cite{Mafalda} correspond to a new class of many-qubit UPBs. It is
still, however, unknown whether these Bell inequalities are tight
for an arbitrary number of parties. In Ref.~\cite{Mafalda},
tightness was verified only for $3\leq n\leq 7$. Here, we prove
tightness for all odd $n\geq 3$.

Let us consider a particular example of games described above (cf.
Sec.~\ref{NC}). Assume that we have $n$ parties and each of them
is given a bit $x_i\in\{0,1\}$, forming a string of settings
$\boldsymbol{x}$. The overall goal is that every party guesses the
next party's bit $x_{i+1}$, thus the name {\it guess your
neighbor's input} (GYNI)~\cite{Mafalda}. Additionally, we demand
that all strings $\boldsymbol{x}$ of settings $x_i$ are randomly
chosen, with equal probabilities, from those which satisfy
\begin{equation}\label{rules}
\left\{
\begin{array}{ll}
x_1\oplus \ldots \oplus x_n=0, \,&\, \mathrm{for}\,\,\mathrm{odd}\,\,n\\[1ex]
x_2\oplus \ldots \oplus x_n=0, \,&\,
\mathrm{for}\,\,\mathrm{even}\,\,n
\end{array}
\right..
\end{equation}
For given NC $p(\boldsymbol{a}|\boldsymbol{x})$, the probability of success in this game is given by
\begin{equation}\label{Tight}
P_{\mathrm{succ}}=\frac{1}{2^{n-1}}\sum_{i}p(\widehat{\boldsymbol{x}}_{i}|\boldsymbol{x}_{i}),
\end{equation}
where for a string of settings $\boldsymbol{x}$, the notation
$\widehat{\boldsymbol{x}}$ stands for the string of settings with
$\widehat{x}_i=x_i$ and the sum goes over all input vectors
$\boldsymbol{x}_i$ whose components obey~(\ref{rules}). (Note that
the index in $\boldsymbol{x}_i$ now enumerates the strings of
settings instead of the components of such a string.) It was shown
in~\cite{Mafalda} that the maximal value of~(\ref{Tight}) over all
classical strategies is
$P_{\mathrm{succ}}^{\mathrm{CC}}=1/2^{n-1}$. This, after simple
algebra, leads us to the Bell inequalities which can be written as
\begin{equation}\label{BellIneqOdd}
\sum_{k=0}^{(n-1)/2} \sum_{1=i_1<\ldots<i_{2k}}^{n}D_{i_1\ldots
i_{2k}}p(\boldsymbol{0}|\boldsymbol{0})\leq 1,
\end{equation}
and
\begin{equation}\label{BellIneqEven}
\sum_{k=0}^{(n-2)/2} \sum_{2=i_1<\ldots<i_{2k}}^{n}\hspace{-0.3cm}D_{i_1\ldots
i_{2k}}[p(\boldsymbol{0}|\boldsymbol{0})+p(0\ldots 01|10\ldots 0)]\leq 1,
\end{equation}
for odd and even $n$, respectively. Here
$\boldsymbol{0}=(0,\ldots,0)$ and $D_{i_1,\ldots,i_k}$ flips
($0\leftrightarrow 1$) input bits and output bits at positions
$i_1,\ldots,i_k$ and $i_1-1,\ldots,i_k-1$ (if $i_1=1$ then
$i_1-1=n$), respectively. We call this inequality GYNI$_n$.

Notice that for $n=3$, Eq.~(\ref{BellIneqOdd}) reproduces the Bell
inequality (\ref{Ex1}), while for $n=4$, Eq.~(\ref{BellIneqEven}) gives
\begin{eqnarray}
&&\hspace{-0.8cm}p(0000|0000)+p(0001|1000)+p(0110|0011)\nonumber\\
&&\hspace{-0.8cm}+p(0111|1011)+p(1010|0101)+p(1011|1101)\\
&&\hspace{-0.8cm}+p(0111|1011)+p(1100|0110)+p(1101|1110)\leq 1\nonumber.
\end{eqnarray}

It was shown in Ref.~\cite{Mafalda} that GYNI$_n$ cannot be
violated by quantum theory, meaning that the maximum of the
probability $P_{\mathrm{succ}}$~(\ref{Tight}) over all quantum
strategies is exactly the same as
$P_{\mathrm{succ}}^{\mathrm{CC}}$, i.e.,
$P_{\mathrm{succ}}^{\mathrm{QC}}=P_{\mathrm{succ}}^{CC}=1/2^{n-1}$.
Moreover, these inequalities are tight for $3\leq n\leq
7$~\cite{Mafalda}, providing the first examples of tight Bell
inequalities with no quantum violation.

It is then interesting to study these inequalities from the point
of view of our construction. Below, we will show that the product
vectors corresponding to the Bell inequalities~(\ref{BellIneqOdd})
and~(\ref{BellIneqEven}) constitute an $n$-qubit UPB. For this
purpose, first notice that the conditional probabilities have two
possible settings inputs and outcomes for each party. This means
that the Hilbert space supporting the product vectors is
$\mathcal{H}=(\mathbbm{C}^2)^{\ot n}$ and at each site we have two
bases, which for simplicity we can take to be the same for all
sites and given as $\mathcal{S}_0=\{\ket{0},\ket{1}\}$ and
$\mathcal{S}_1=\{\ket{e},\ket{\overline{e}}\}$ with $\ket{e}$
being different from $\ket{0}$ and $\ket{1}$. Clearly, one can
consider bases that are different at each site leading to more
general UPBs.

Let now $V$ denote the unitary operator mapping $\mathcal{S}_0$ to
$\mathcal{S}_1$, that is $V\ket{0}=\ket{e}$ and
$V\ket{1}=\ket{\overline{e}}$; we write $V_i$ for an application
of $V$ on the $i$th qubit. Then, following the rules described
above, one sees that the product vectors corresponding
to~(\ref{BellIneqOdd}) read, with $\sigma_i$ standing for the
Pauli operator $\sigma_x|0\rangle=|1\rangle$ and
$\sigma_x|1\rangle=|0\rangle$ acting on qubit $i$,
\begin{align}
\begin{split}
\label{vectorsOdd}
V_{i_1}\ldots
V_{i_k}\sigma_{i_{1}-1}\ldots&\sigma_{i_{k}-1}\ket{0}^{\ot
n},\\
i_{1}<\ldots<i_{k}\in\{1,\ldots,n\},\quad
 k = 0, & 2,4,\ldots,n-1,
\end{split}
\end{align}
while those corresponding to~(\ref{BellIneqEven}) read
\begin{align}
\begin{split}
\label{vectorsEven}
V_{i_1}\ldots V_{i_k} &\sigma_{i_{1}-1}\ldots \sigma_{i_{k}-1}\ket{0}^{\ot
n},\\
V_1V_{i_1}\ldots V_{i_k} &\sigma_{i_{1}-1}\ldots\sigma_{i_{k}-1}\sigma_n\ket{0}^{\ot
n},\\
i_{1}<\ldots<i_{k}\in \{2,&\ldots,n\},\quad k=0, 2,4,\ldots,n-2.
\end{split}
\end{align}
Let us denote the set of these vectors by $U_n$. We will show in
Theorem~\ref{GYNIUPB} that $U_n$ is a UPB.

In the particular cases of $n=3$, Eq.~(\ref{vectorsOdd}) give the
Shifts UPB~(\ref{Schifts}). For $n=4$, Eq.~(\ref{vectorsEven})
gives
\begin{eqnarray}
U_4&=&\{\ket{0000},\ket{1\overline{e}e0},\ket{e001},\ket{\overline{e}\overline{e}e1},\nonumber\\
&&\ket{01\overline{e}e},
\ket{1e1e},\ket{e1\overline{e}\overline{e}},\ket{\overline{e}e1\overline{e}}\}.
\end{eqnarray}
A direct check shows that $U_4$ is also a UPB in
$(\mathbbm{C}^2)^{\ot 4}$. Below we show that this is the case for
any $n$. To this end, first let us state two simple facts.

\begin{lem}\label{Lemma1}
Let $U$ be an $n$-qubit UPB. Then, the set $\widetilde{U}$
obtained from $U$ by substituting some local basis at site $i$ by
some other basis different from the other local bases at this site
is also a UPB.
\end{lem}
\begin{proof}
The proof is trivial. As said before (cf. Sec.~\ref{constr}),
every $n$-qubit UPB has the property (P), or, in other words, the
property of being UPB in qubit Hilbert spaces is independent of
the choice of local bases. Hence, by replacing any local basis
with any other independent basis we just get another UPB.
\end{proof}

\begin{lem}\label{Lemma2}
Let $U_1$ and $U_2$ be two
$n$-qubit UPBs. Assume that both can be divided into $k$ subsets
$U_1^{(i)}$ and $U_{2}^{(i)}$ $(i=1,\ldots,k)$ obeying the
following orthogonality rules
\begin{equation}\label{OrtCond}
U_1^{(i)}\perp U_{2}^{(j)}\qquad (i,j=1,\ldots,k;\;i\neq j).
\end{equation}
Then the set of vectors
\begin{equation}\label{NextUPB}
\begin{array}{cc}
\ket{0}\ot U_1^{(1)}, \quad & \quad \ket{1}\ot U_2^{(1)},\\[1ex]
\ket{e_1}\ot U_1^{(2)}, \quad & \quad \ket{\overline{e}_1}\ot U_2^{(2)},\\[1ex]
\vdots & \vdots\\[1ex]
\ket{e_{k-1}}\ot U_1^{(k)}, \quad & \quad \ket{\overline{e}_{k-1}}\ot U_2^{(k)},
\end{array}
\end{equation}
where $\{\ket{0},\ket{1}\}$ and
$\{\ket{e_i},\ket{\overline{e}_i}\}$ $(i=1,\ldots,k-1)$ are
different bases in $\mathbbm{C}^2$, constitutes
a $(n+1)$-qubit UPB.
\end{lem}

\begin{proof}
It follows directly from the assumptions that all vectors in this
set are mutually orthogonal. Now assume that the
vectors~(\ref{NextUPB}) are not a UPB, but that there exists a
product vector $\ket{\psi}\in(\mathbbm{C}^{2})^{\ot (n+1)}$
orthogonal to all of them. Writing this vector as
$\ket{\psi}=\ket{x}\ot\ket{\widetilde{\psi}}$ with
$\ket{\widetilde{\psi}}\in(\mathbbm{C}^{2})^{\ot n}$ and $\ket{x}$
being a one-qubit vector, $\ket{x}$ may or may not belong to one
of the local bases $\{\ket{0},\ket{1}\}$,
$\{\ket{e_i},\ket{\overline{e}_i}\}$ $(i=1,\ldots,k-1)$. If it
does, $\ket{\widetilde{\psi}}$ has to be orthogonal to one of the
sets $U_1$ or $U_2$, while if it does not,
$\ket{\widetilde{\psi}}$ must be orthogonal to both of them. In
either case, this contradicts the assumption that $U_1$ and $U_2$
are UPBs.
\end{proof}

\begin{thm}
\label{GYNIUPB} The product vectors given in
Eqs.~(\ref{vectorsOdd}) and~(\ref{vectorsEven}) constitute a
$2^{n-1}$-element UPB in $(\mathbbm{C}^2)^{\ot n}$.
\end{thm}

\begin{proof}
We use induction on $n$. For the base case $n=3$,
$U_3=\mathcal{U}_{\mathrm{Shifts}}$ is already known to be a UPB.

For the induction step, we partition $U_{n+1}$ into the four subsets
\begin{equation}\label{n+1}
\begin{array}{cc}
\ket{0}\ot U_{n}^{(0)}, \quad & \quad \ket{1}\ot
\widetilde{U}_{n}^{(1)},\\[1ex]
\ket{\overline{e}} \ot U_{n}^{(1)}, \quad & \quad \ket{e}\ot
\widetilde{U}_n^{(0)},
\end{array}
\end{equation}
for certain $n$-qubit sets of vectors $U_{n}^{(0)}$,
$U_{n}^{(1)}$, $\widetilde{U}_n^{(0)}$ and
$\widetilde{U}_{n}^{(1)}$. We now establish the assertions of
Lemma~\ref{Lemma2}. The first observation is that
$$
U_n = U_n^{(0)} \cup U_n^{(1)},
$$
which follows from the definitions~(\ref{vectorsOdd})
and~(\ref{vectorsEven}) by considering the cases of even and odd
$n$ separately. This $U_n$ is known to be an $n$-qubit UPB by the
induction assumption.

Similarly, it can be shown that
$$
\widetilde{U}_n = \widetilde{U}_n^{(0)} \cup \widetilde{U}_{n}^{(1)}
$$
is the set which for odd $n$ one obtains from $U_n$ by replacing
$|0\rangle\leftrightarrow |1\rangle$ and
$|e\rangle\leftrightarrow\ket{\overline{e}}$ in the last qubit,
while for even $n$, by replacing $|0\rangle\leftrightarrow
|1\rangle$ and $|e\rangle\leftrightarrow\ket{\overline{e}}$ in the
second to last qubit and
$|0\rangle\leftrightarrow|\overline{e}\rangle$ and
$|1\rangle\leftrightarrow\ket{e}$ in the last one. In the
following, we denote this operation by $F_n$. An application of
Lemma~\ref{Lemma1} shows that $\widetilde{U}_n$ is also a UPB.

In order to complete the proof by an application of
Lemma~\ref{Lemma2}, we need to show the orthogonality relations
$$
U_n^{(0)}\perp \widetilde{U}_n^{(0)} ,\qquad U_n^{(1)}\perp \widetilde{U}_n^{(1)} .
$$
Since the proof of the second relation is analogous, we only prove
the first. Let us consider the cases of odd $n$ and even $n$
separately and start with odd $n$. As follows from
Eq.~(\ref{vectorsOdd}), the last qubit of elements of
$U_{n}^{(0)}$ is either $\ket{0}$ or $\ket{e}$, meaning that
orthogonality of any two vectors $\ket{\psi_j},\ket{\psi_k}\in
U_{n}^{(0)}$ $(j\neq k)$ comes from one of the first $n-1$ qubits.
Now, on the one hand, $F_n$ acts on the last qubit, and therefore
$\ket{\psi_j}\perp F_n\ket{\psi_k}\in\widetilde{U}_n^{(0)}$ for
any $k\neq j$. On the other hand, as $F_n$ replaces the last qubit
with an orthogonal one, $\ket{\psi_j}\perp
F_n\ket{\psi_j}\in\widetilde{U}_n^{(0)}$. These two facts together
imply that an arbitrary vector $\ket{\psi_j}\in U_n^{(0)}$ is
orthogonal to any elements of $\widetilde{U}_n^{(0)}$, hence
$U_n^{(0)}\perp \widetilde{U}_n^{(0)}$.

In the case of even $n$, the operation $F_n$ acts on the last two
qubits. Nevertheless, it follows from Eq.~(\ref{vectorsEven}) [cf.
Eq.~(\ref{BellIneqEven})] that the orthogonality of any two
vectors $\ket{\psi_j},\ket{\psi_k}\in U_n^{(0)}$ $(j\neq k)$ comes
from one of the first $n-2$ qubits. Precisely, as already stated,
they cannot be orthogonal at the last qubit. Assume then that they
are orthogonal at the second to last one. In this case, the
corresponding conditional probabilities have the same input bits
at this site but different at the last one. Taking into account
Eq.~(\ref{rules}), this means that at some other site $k$
($1<k<n-2$), these conditional probabilities have again different
inputs, but at site $k-1$ their inputs are equal. Consequently,
due to the fact that the outputs of these conditional
probabilities are inputs shifted to the left, they must be
``orthogonal'' at one of the first $n-2$ sites. Hence, the
orthogonality of the corresponding vectors $\ket{\psi_j}$ and
$\ket{\psi_k}$ comes from one of the first $n-2$ qubits. Having
this, we just apply the reasoning developed for odd $n$,
concluding that $U_n^{(0)}\perp\widetilde{U}_n^{(0)}$.
\end{proof}

\begin{thm}
\label{gynitight}
For odd $n$, the $\mathrm{GYNI}_n$ Bell inequalities~(\ref{BellIneqOdd}) are tight.
\end{thm}

\begin{proof}
The proof is moved to the appendix.
\end{proof}

Let us finally notice that, as shown in Ref.~\cite{Mafalda}, the
maximal NC value of the GYNI Bell inequalities is, irrespectively
of the number of parties, always upper bounded by 2. In general,
however, one is able to find Bell inequalities constructed from
UPBs having arbitrarily large nonsignalling violation. Let us
consider for this purpose an $n$-partite UPB $U$ and its
associated Bell inequality~(\ref{BellIneqEqual}). It is violated
by the witness~(\ref{witness}) constructed from $U$. We write
$\beta>1$ for the value of the Bell inequality on $W$.

Then it follows from Ref.~\cite{UPBhuge} that the $k$-fold tensor
product of $U$, i.e., $U^{\ot k}$, is a $kn$-partite UPB. Via our
construction, this new UPB turns into a $kn$-partite Bell
inequality~(\ref{BellIneqEqual}), which is the $k$-fold product of
the $n$-partite Bell inequality associated to $U$. Nevertheless,
its value on the $kn$-partite witness $W^{\ot k}$ amounts to
$\beta^{k}$, which in the limit of large $k$ becomes arbitrarily
large. Therefore, the NC violation of Bell inequalities
constructed from UPBs is unbounded and may grow exponentially in
the number of parties.

\section{New tight Bell inequalities with no quantum violation from UPBs}\label{New}

The purpose of this section is to provide new examples of tight
Bell inequalities without quantum violation from UPBs. In order to
do this, we first construct new examples of four-qubit UPBs, and
later map them into Bell inequalities. Then, we discuss simple
methods allowing one to extend $n$-partite UPBs to $(n+1)$-partite
ones. Interestingly, these methods, on the level of Bell
inequalities, allow us to construct $(n+1)$-partite Bell
inequalities from $n$-partite ones; our computations show that
this construction frequently produces tight Bell inequalities.

\subsection{The four-qubit UPBs}

Recall that in the case of two qubits, there are no UPBs. Moving
to three qubits, it is known that there is only one class of
UPBs~\cite{Bravyi}, which we have found in example~\ref{Ex1} to
give rise to the GYNI$_3$ inequality, the only tight and
nontrivial Bell inequality with no quantum violation in the
tripartite scenario with two settings and two
outcomes~\cite{Mafalda}. Thus, the first still unexplored case to
be analyzed is the four-qubit case.

In order to search for four-qubit UPBs, we used a brute-force
numerical procedure. In this way, we found many new examples of
UPBs, some of which allowed us to construct nontrivial tight Bell
inequalities. In particular, we have found four-qubit UPBs
providing tight Bell inequalities with two settings per site which
are inequivalent to GYNI$_4$~(\ref{BellIneqEven}). We have also
obtained some Bell inequalities with two and three observables at
some sites.

In Table~\ref{TableI}, we have collected some of these UPBs.
Table~\ref{TableII} shows the corresponding Bell inequalities and
their maximal violations by NC, as computed by linear programming.
All the first nine Bell inequalities are tight. For the sake of
completeness, we also present a four-qubit UPB, denoted
$\mathcal{U}_{10}$, leading to a nontight Bell inequality. Notice
that
$\mathcal{U}_{10}=\{\ket{0}\ot\mathcal{U}_{\mathrm{Shifts}},\ket{1}\ot\mathcal{U}_{\mathrm{Shifts}}\}$,
i.e., it is just a tensor product of the standard basis in
$\mathbbm{C}^2$ and the three-qubit Shifts UPB~(\ref{Schifts}),
which implies that the associated Bell inequality has only one
observable at the first site.

In Tables~\ref{TableI} and~\ref{TableII}, by $(i,j,k,l)$ we denote
the number of bases/observables per site (recall that we always
have two outcomes) and by $\{\ket{0},\ket{1}\}$,
$\{\ket{e},\ket{\overline{e}}\}$, and
$\{\ket{f},\ket{\overline{f}}\}$ we denote the independent local
bases (in the sense of the property (P)).
\begin{widetext}

\begin{table}[h!]
\begin{tabular}{c|l}
\hline
$\#$ bases per site & UPB\\
\hline\hline
(2,2,2,2) &\, $\mathcal{U}_1=\{\ket{0000},\ket{1e\overline{e}0},
\ket{e\overline{e}10},\ket{\overline{e}1ee},\ket{0001},
\ket{01\overline{e}1},\ket{1\overline{e}0\overline{e}},\ket{0011},\ket{1011}\}$
\\

(2,2,2,2) &\,
$\mathcal{U}_2=\{\ket{0000},\ket{\overline{e}e1e},\ket{e1e1},
\ket{\overline{e}11\overline{e}},\ket{ee\overline{e}1},
\ket{1\overline{e}\overline{e}e},\ket{10e\overline{e}},
\ket{\overline{e}10\overline{e}},\ket{e1e0}\},$\\

(2,2,2,2) &\,
$\mathcal{U}_3=\{\ket{0000},\ket{1e\overline{e}0},\ket{e\overline{e}10},\ket{\overline{e}1ee},\ket{0001},
\ket{0011},\ket{1001},\ket{1011},\ket{010\overline{e}},
\ket{11\overline{e}1}\}$\\

(2,2,2,1) & \, $\begin{array}{l}
\hspace{-0.12cm}\mathcal{U}_4=\{\ket{0000},\ket{1\overline{e}e0},\ket{e1\overline{e}10},\ket{\overline{e}e10},\ket{0001},\ket{0011},
\ket{0101},\ket{0111},\ket{1001},\ket{1011},\ket{1101},\ket{1111}\}\nonumber\\
\hspace{0.32cm}=\{\mathcal{U}_{\mathrm{Shifts}}\ot\ket{0},\mathcal{B}\ot\ket{1}\},
\end{array}$\\

\hline

(2,2,2,3) &\,
$\mathcal{U}_5=\{\ket{0000},\ket{1eee},\ket{e1\overline{e}f},\ket{0ee1},
\ket{101\overline{e}},\ket{1\overline{e}0\overline{f}},\ket{\overline{e}\overline{e}1e}\},$\\

(2,2,2,3) & \, $\mathcal{U}_6=\{\ket{0000},\ket{1e\overline{e}e},
\ket{e\overline{e}1\overline{e}},\ket{\overline{e}1ef},
\ket{0e\overline{e}1},\ket{1e\overline{e}\overline{e}},
\ket{e1ee},\ket{\overline{e}\overline{e}1\overline{f}}\},$\\

(2,2,2,3) &\, $\mathcal{U}_7=\{\ket{0000},\ket{
\overline{e}ee1},\ket{e11e},
\ket{\overline{e}1\overline{e}\overline{f}},\ket{1000},\ket{e001},
\ket{e10\overline{e}},\ket{e010},\ket{e011},
\ket{e11\overline{e}},\ket{\overline{e}\overline{e}1f},\ket{e10e}\},$\\

\hline

(2,2,3,3) &\,
$\mathcal{U}_8=\{\ket{0000},\ket{e1e\overline{e}},
\ket{1e\overline{e}f},\ket{\overline{e}\overline{e}\overline{f}1},
\ket{e01\overline{f}},\ket{01fe}\}$\\

(2,2,3,3) &\,
$\mathcal{U}_9=\{\ket{0000},\ket{1eee},\ket{e\overline{e}1f},\ket{10\overline{e}\overline{f}},\ket{0ef1},
\ket{01\overline{f}\overline{f}},\ket{1\overline{e}0f},\ket{1ee
\overline{e}},\ket{1\overline{e}e\overline{f}},
\ket{1e\overline{e}f},\ket{11\overline{e}\overline{f}}
,\ket{\overline{e}\overline{e}1f}\}$\\

\hline\hline

(1,2,2,2) &\,
$\mathcal{U}_{10}=\{\ket{0000},\ket{01\overline{e}e},\ket{0e1\overline{e}},\ket{0\overline{e}e1},
\ket{1000},\ket{11\overline{e}e},\ket{1e1\overline{e}},\ket{1\overline{e}e1}\}=
\{\ket{0}\ot\mathcal{U}_{\mathrm{Shifts}},\ket{1}\ot\mathcal{U}_{\mathrm{Shifts}}\}$

\end{tabular}
\caption{New four-qubit UPBs $\mathcal{U}_i$ $(i=1,\ldots,9)$
leading to tight nontrivial Bell inequalities with no quantum
violation. Notice, that $\mathcal{U}_4$ takes the simple form
$\{\mathcal{U}_{\mathrm{Shifts}}\ot\ket{0},\mathcal{B}\ot\ket{1}\}$
with $\mathcal{B}$ denoting the standard basis in
$(\mathbbm{C}^{2})^{\ot 3}$. For completeness we also present a
UPB $\mathcal{U}_{10}$ for which the associated Bell inequality is
not tight. The first column contains the number of independent
bases per site. One can check by hand that each $\mathcal{U}_i$ is
a set of mutually orthogonal product vectors and that there is no
other product vector orthogonal to all of them.}~\label{TableI}
\end{table}
%

%

\begin{table}[h!]
\begin{tabular}{c|c|l}
  \hline
UPB    & $\#$ observables per site & Bell inequality
\\
  \hline\hline


{\scriptsize $\mathcal{U}_1$} & (2,2,2,2)

&

{\scriptsize $\begin{array}{l}
p(0000|0000)+p(1010|0110)+p(0110|1100)
+p(1100|1011)+p(0001|0000)+p(0111|0010)\\
+p(1101|0101)+p(0011|0000)+p(1011|0000)\leq 1
\end{array}$}

\\


{\scriptsize $\mathcal{U}_2$} & (2,2,2,3)

&

{\scriptsize
$\begin{array}{l}p(0000|0000)+p(1010|1101)+p(0101|1010)
+p(1111|1001)+p(0011|1110)+p(1110|0111)\\
+p(1001|0011)+p(1101|1001)+p(0100|1010)\leq 1\end{array}$}
\\


{\scriptsize $\mathcal{U}_3$} & (2,2,2,2)

&

{\scriptsize $\begin{array}{l}p(0000|0000)+p(1010|0110)+p(0110|1100)+p(1100|1011)+p(0001|0000)+p(0011|0000)\\
+p(1011|0000)+p(1001|0000)+p(0101|0001)+p(1111|0010)\leq
1
\end{array}$}

\\


{\scriptsize $\mathcal{U}_4$} & (2,2,2,2)

&

{\scriptsize
$\begin{array}{l}p(0000|0000)+p(1100|0110)+p(0110|1010)+p(1010|1100)+p(1101|0000)+p(0001|0000)\\
+p(0011|0000)+p(0101|0000)+p(0111|0000)+p(1001|0000)+p(1011|0000)+p(1111|0000)\leq1
\end{array}$}

 \\ \hline


 {\scriptsize $\mathcal{U}_5$} & (2,2,2,3)
&

{\scriptsize$\begin{array}{l}p(0000|0000)+p(1000|0111)+
p(0110|1012)+p(0001|0110)+p(1011|0001)+p(1101|0102)\\
+p(1110|1101)\leq 1\end{array}$}
\\


{\scriptsize $\mathcal{U}_6$} & (2,2,2,3)

&

{\scriptsize
$\begin{array}{l}p(0000|0000)+p(1010|0111)+p(0111|1101)
+p(1100|1012)+p(0011|0110)+p(1011|0111)\\
+p(0100|1011)+p(1111|1102)\leq 1\end{array}$}
\\


{\scriptsize $\mathcal{U}_7$} & (2,2,2,3)

&

{\scriptsize $\begin{array}{l}
p(0000|0000)+p(1001|1110)+p(0110|1001)+p(1111|1012)+p(1000|0000)+p(0001|1000)\\
+p(0101|1001)+p(0010|1000)+p(0011|1000)+p(0111|1001)+p(1110|1102)+p(0100|1001)\leq
1
\end{array}$}
\\

\hline

{\scriptsize $\mathcal{U}_8$} & (2,2,3,3)

&

{\scriptsize $
p(0000|0000)+p(0101|1011)+p(1010|0112)+p(1111|1120)+p(0011|1002)+p(0100|0021)\leq1$}
\\


{\scriptsize $\mathcal{U}_9$} & (2,2,3,3)

&

{\scriptsize$\begin{array}{l}
p(0000|0000)+p(1000|0111)+p(0110|1102)+p(1011|0012)+p(0001|0120)+p(0111|0022)\\
+p(1100|0102)+p(1001|0111)+p(1101|0112)+p(1010|0112)+p(1111|0012)+p(1110|1102)\leq
1
\end{array}$}
\\

\hline\hline

{\scriptsize $\mathcal{U}_{10}$} & (1,2,2,2)

&

{\scriptsize $\begin{array}{l}
p(0000|0000)+p(0110|0011)+p(0011|0101)+p(0101|0110)+p(1000|0000)+p(1101|0011)\\
+p(1011|0101)+p(1101|0110)\leq1
\end{array}$}
\\

  \hline
\end{tabular}
\caption{Tight four-partite Bell inequalities with no quantum
violation (third column) constructed from the UPBs $\mathcal{U}_i$
$(i=1,\ldots,9)$ given in Table~\ref{TableI}. The fourth one is a
lifting of GYNI$_3$~(\ref{Ex1}) (see Sec.~\ref{more}). For
completeness, we also present a nontight Bell inequality,
constructed from the UPB $\mathcal{U}_{10}$. Noticeably, in this
case, the first observer has only one observable at his disposal.
The second column contains the number of observables per site.
Neither of the inequalities is equivalent to the GYNI's Bell
inequalities~(\ref{BellIneqEven}). Moreover, all have the same
maximal NC violation of $4/3$.}\label{TableII}
\end{table}

\end{widetext}

\subsection{Going to more parties}
\label{more}

Here we discuss some methods for extending $n$-partite Bell
inequalities, which are assumed to be associated to $n$-qubit
UPBs, to $(n+1)$-partite ones.

\paragraph*{Method 1.} Let $U_i$ $(i=1,2)$ be two, in general
different, $n$-partite UPBs. Then, consider the $(n+1)$-qubit set
of vectors $\{U_1\ot\ket{0},U_2\ot\ket{1}\}$. One immediately
checks that the latter is an $(n+1)$-partite UPB and therefore
allows one to construct a nontrivial $(n+1)$-partite Bell
inequality with no quantum violation. Nevertheless, the additional
party has only one measurement at his choice, and, moreover, the
obtained Bell inequality does not have to be tight, even if the
ones associated to $U_i$ were tight. A particular example of such
UPB (the associated Bell inequality is not tight) is
$\mathcal{U}_{10}$ (see Table~\ref{TableI}).

In order to obtain a tight Bell inequality in this way, one has to
replace one of the UPBs, say $U_2$, with a full basis
$\mathcal{B}$ (say the standard one) in $(\mathbbm{C}^2)^{\ot n}$.
Then, the set
\begin{equation}\label{niewiem}
\{U_1\ot\ket{0},\mathcal{B}\ot\ket{1}\}
\end{equation}
is again a UPB, however, this time it leads to a tight Bell
inequality, provided that the Bell inequality associated to $U_1$
is tight. This is because, on the level of Bell inequalities, such
a construction corresponds to the lifting to one more observer
studied in Ref.~\cite{Lifting}. It is proven there that such a
procedure always gives a tight Bell inequality if the initial one
was tight. Consequently, if the UPB $U_1$ corresponds to a tight
$n$-partite Bell inequality, the UPB~(\ref{niewiem}) always leads
to a tight $(n+1)$-partite Bell inequality.

A simple example of a Bell inequality constructed in this way is
the one associated to the UPB $\mathcal{U}_4$ (see
Table~\ref{TableII}). A direct check allows one to conclude that
$\mathcal{U}_4=\{\mathcal{U}_{\mathrm{Shifts}}\ot\ket{0},\mathcal{B}\ot\ket{1}\}$,
where $\mathcal{U}_{\mathrm{Shifts}}$ is given by
Eq.~(\ref{Schifts}), while $\mathcal{B}$ denotes the standard
basis in $(\mathbbm{C}^2)^{\ot 3}$. Notice that with the aid of
the nonsignalling conditions~(\ref{nonsignalling}), the associated
Bell inequality can be rewritten as
\begin{eqnarray}
&&p(0000|0000)+p(1100|0110)+p(0110|1010)\nonumber\\
&&+p(1010|1100)\leq p_{A_4}(0|0),
\end{eqnarray}
where $p_{A_4}(0|0)=\sum_{a_1,a_2,a_3}p(a_1a_2a_30|x_1x_2x_30)$
with arbitrary binary $x_i$ $(i=1,2,3)$.

This method, however, produces a Bell inequality with a single
observable at the new site. Now we describe a method which
allows to add an observer with more observables at his choice.

\paragraph*{Method 2.} Another simple and general method follows
from lemmas~\ref{Lemma1} and~\ref{Lemma2} and resembles the
methods already used in the literature (see
e.g.~Ref.~\cite{metody}).

Before we make it explicit, let us recall that if two given
$n$-qubit UPBs $U_i$ $(i=1,2)$ obey the assumptions of
lemma~\ref{Lemma2}, then the set of vectors~(\ref{NextUPB}) forms
an $(n+1)$-qubit UPB. On the level of Bell inequalities, this
procedure allows us to construct an $(n+1)$-partite Bell
inequality from two $n$-partite ones associated to $U_i$
$(i=1,2)$. One way to chose $U_2$ is to start from a UPB $U_1$ and
then apply Lemma~\ref{Lemma1}. Moreover, as we will see below, it
is always possible to do this in such a way that $U_1$ and $U_2$
satisfy the assumptions of lemma~\ref{Lemma2}. These two facts
give a quite general method of extending nontrivial $n$-partite
Bell inequalities with no quantum violation to $(n+1)$-partite
ones. A particular example of such an approach is the recursive
method relating GYNI$_{n+1}$ to GYNI$_n$ described in the previous
section. Here, we show another simple way to obtain a UPB $U_2$
from a given UPB $U_1$ such that~(\ref{OrtCond}) holds and
illustrate it with examples of five-partite tight Bell
inequalities.

Consider an $n$-qubit UPB $U_1$ which has $k$ different bases
$\mathcal{S}_j^{(i)}$ $(j=1,\ldots,k)$ at the $i$th site. By
replacing, at the site $i$, all local vectors by their orthogonal
complements, we turn $U_1$ into another set of vectors $U_2$,
which, due to Lemma~\ref{Lemma1}, is also a UPB. Then, we divide
both UPBs $U_i$ $(i=1,2)$ into $k$ subsets $U_i^{(j)}$
$j=1,\ldots,k$, each having at the $i$th site vector from
$\mathcal{S}^{(i)}_j$. One immediately sees that the sets
$U_1^{(j)}$ and $U_2^{(j)}$ $(j=1,\ldots,k)$ obey the assumptions
of lemma~\ref{Lemma2}. Therefore, the vectors~(\ref{NextUPB}) form
an $(n+1)$-partite UPB.

This method translates to Bell inequalities and allows one to
obtain an $(n+1)$-partite Bell inequality from an $n$-partite one
of the form~(\ref{BellIneqEqual}). More importantly, at least in
some cases this procedure preserves tightness. We applied it to
some of the Bell inequalities presented in Table~\ref{TableII} and
the resulting tight Bell inequalities are collected in
Table~\ref{TableIII}.
%

%
\begin{table*}[]
\begin{tabular}{c|c|l}
  \hline
UPB    & scenario & Bell inequality \\
  \hline\hline


{\scriptsize $(2,\mathcal{U}_1)$} & (2,2,2,2,2) &

{\scriptsize $\begin{array}{l}
p(00000|00000)+p(00001|00000)+p(00011|00000)+p(01011|00000)+p(00111|00010)+p(01100|01011)\\
+p(01101|10101)+p(01010|10110)+p(00110|11100)+p(10100|00000)+p(10101|00000)+p(10111|00000)\\
+p(11111|00000)+p(10011|00010)+p(11000|01011)+p(11001|10101)+p(11110|10110)+p(10010|11100)\leq
1
\end{array}$}

\\ \hline


{\scriptsize $(1,\mathcal{U}_6)$} & (2,2,2,2,3) &

{\scriptsize $\begin{array}{l}
p(00000|00000)+p(11000|00000)+p(00011|00110)+p(11011|00110)+p(01010|00111)+p(10010|00111)\\
+p(01011|00111)+p(10011|00111)+p(00111|11101)+p(11111|11101)+p(00100|11011)+p(11100|11011)\\
p(01100|11012)+p(10100|11012)+p(01111|11102)+p(10111|11102)\leq 1
\end{array}$}

 \\ \hline


 {\scriptsize $(4,\mathcal{U}_8)$} & (2,2,3,3,3) &

{\scriptsize$\begin{array}{l}
p(00000|00000)+p(11110|11200)+p(01010|10111)+p(01000|00211)+p(10100|01122)+p(00110|10022)\\
p(00011|00000)+p(11101|11200)+p(01001|10111)+p(01011|00211)+p(10111|01122)+p(00101|10022)\leq
1
\end{array}$} \\ \hline

  \hline
\end{tabular}
\caption{Five-partite Bell inequalities obtained by applying our
method to some of the UPB listed in Table~\ref{TableI}. The first
column contains the UPB used to construct the Bell inequality
together with the number of party to which we applied our method
(in the last case we put the fifth party at the end), while the
second column contains the number of observables per site. Our
calculations have shown all these Bell inequalities to be
tight.}\label{TableIII}
\end{table*}


\section{Conclusion}
\label{Conclusion}

Let us shortly summarize the obtained results and outline
possible directions for further research.

In this work, we have investigated in more detail the relation
between UPBs and nontrivial Bell inequalities with no quantum
violation which had been discovered recently in Ref.~\cite{my}. We
have restricted ourselves to qubit Hilbert spaces, since in this
case the sets of product vectors always have the property (P).
First, we have proven that a Bell inequality associated to any
$n$-qubit set of orthogonal product vectors $\mathcal{S}$ that can
be completed to a full basis is trivial in the sense that it
cannot be violated by any nonsignalling correlations. This result,
on the one hand, significantly supplements the characterization
done in Ref.~\cite{my}. On the other hand, it adds additional
weight to the significance of the concept of UPB in our
construction. The only nontrivial Bell inequalities that can be
obtained by our method from sets of product vectors are precisely
those associated to UPBs or sets of vectors that can be extended
only to UPBs.

Secondly, and more importantly, we have provided new examples of
\textit{tight} Bell inequalities with no quantum violation
constructed from UPBs. So far, the only known examples have been
the GYNI's Bell inequalities \cite{Mafalda}; whether UPBs can be
used for the construction of tight Bell inequalities remained an
open question. Here we have investigated the simplest nontrivial
case (four qubits) and found several examples. Finally, we have
presented methods allowing one to extend an $n$-qubit UPB to an
$(n+1)$-qubit one. On the level of Bell inequalities, these
methods enable to lift an $n$-partite Bell inequality associated
to a UPB to an $(n+1)$-partite one, and in some cases they
preserve tightness.

Third, we have proven tightness of the GYNI's
Bell inequalities for an arbitrary odd number of parties.

Nevertheless, many questions remain unanswered. A similar analysis
is missing in the case when the local dimensions are higher than
two. As already stated in Ref.~\cite{my}, theorems~\ref{theorem1}
and~\ref{theorem3} remain valid in arbitrary Hilbert spaces. It
is, however, unclear if in higher-dimensional Hilbert spaces the
only sets of mutually orthogonal product vectors that lead to
nontrivial Bell inequalities with no quantum violation are UPBs or
those that can be completed only to UPBs.

More importantly, it would be of interest to understand when our
construction leads to tight Bell inequalities. Although we already
have many examples of such inequalities, there still exist UPBs
that do not correspond to tight inequalities. It is unclear which
property of a UPB guarantees the tightness of the associated Bell
inequality. Since the only known examples of tight Bell
inequalities with no quantum violation are those obtained from
many-qubit UPBs, it is reasonable to conjecture that all such Bell
inequalities arise from UPBs.

\begin{acknowledgements}
Discussion with J. Stasi\'nska are acknowledged. This work was
supported by EU projects AQUTE, NAMEQUAM, Q-Essence and QCS, ERC
Grants QUAGATUA and PERCENT, Spanish MINCIN projects FIS2008-00784
(TOQATA), FIS2010-14830 and CHIST-ERA DIQIP, and UK EPSRC. R. A.
is supported by the Spanish MINCIN through the Juan de la Cierva
program. M. K. and M. K. acknowledge ICFO for kind hospitality.
\end{acknowledgements}

\appendix
\label{app}

\section{Tightness of GYNI's Bell inequalities}

Before turning to the proof of Theorem~\ref{gynitight}, we need some preparation.

In order to study whether a given Bell inequality defines a facet
of the polytope of CC, it needs to be determined which of the
local deterministic points (which we also call
\textit{strategies}; they are the extremal points of the polytope
of CC) saturate the Bell inequality, i.e.~which of them satisfy
the inequality with equality. In the Bell scenario with $n$
parties, one input bit per party and one output bit per party, the
possible local strategies are the following:
\begin{enumerate}
\item The constant strategy outputting $0$, independently of the input. By abuse of notation, we write $0$ for this strategy.
\item Similarly, the constant strategy always outputting $1$, for which we write $1$.
\item The ``identity'' strategy, under which the output equals the input. We denote this strategy by $i$.
\item The ``flip'' strategy, under which the output equals the negation of the input. We denote this strategy by $f$.
\end{enumerate}
An $n$-partite strategy then is defined by a string
$\boldsymbol{s}=[s_1\ldots s_n]$, with $s_j\in\{0,1,i,f\}$. In
order to clarify notation, we use square brackets to emphasize
that the string denotes a strategy. For example, \beq
\label{exstrategy} [i0f] \eeq stands for the strategy where the
first party outputs their input, the second party always $0$, and
the third party the negation of their input. Here and in the
following, we always regard the index $j$ as defined modulo $n$,
so that $s_{j+n} = s_j$.

While $i$ and $f$ are functions which map a bit to a bit, we may
think of $0$ also as a function, or alternatively as the value of
a bit; likewise for $1$. It will be clear from the context which
point of view is required.

Given a strategy $\boldsymbol{s}$ such that $s_k\in\{0,1\}$ for
some $k$, we define its \textit{evaluation}, denoted by
$\boldsymbol{s}'=[s'_1\ldots s'_n]$, by setting $s'_k=s_k$ and
recursively taking $s'_j = s_j(s'_{j-1})$, starting at $j=k+1$.
Here, $s_j(s'_{j-1})$ is the bit obtained by applying the function
$s_j$ to the bit $s'_{j-1}$,. If there are several $k$ for which
$s_k\in\{0,1\}$, then the evaluation does not depend on the choice
of $k$. Evaluation assigns to any strategy containing at least one
numerical entry a strategy consisting only of numerical entries.
For example, $[i0f]'=[101]$.

\begin{lem}
A strategy $\boldsymbol{s}$ saturates $\mathrm{GYNI}_n$
for odd $n$ iff
\begin{enumerate}
\item $s_k\in\{0,1\}$ for some $k$, and the evaluation $\boldsymbol{s}'$ has even parity, or
\item $s_k\in\{i,f\}$ for all $k$, and the number of $f$'s in $\boldsymbol{s}$ is even.
\end{enumerate}
\end{lem}

\begin{proof}
By definition of GYNI$_n$, the strategy $\boldsymbol{s}$ saturates
whenever there is a string of settings $\boldsymbol{x}$ having
even parity~(\ref{rules}) such that $x_{j+1} = s_j(x_j)$. When
$s_k\in\{0,1\}$ for some $k$, this means that $x_{k+1} = s_k$,
which can then be generalized to $x_{j+1} = s'_j$ for all $j$
using $x_{j+1} = s_j(x_j)$ and the definition of
$\boldsymbol{s'}$. Therefore, $\boldsymbol{x}$ has even parity if
and only if $\boldsymbol{s}'$ does.

If, on the other hand, $s_k\in\{i,f\}$ for all $k$, then we can use the equation
$$
x_1 = s_n(x_n) = (s_n\circ s_{n-1})(x_{n-1}) = \ldots =
(s_n\circ\ldots \circ s_1)(x_1)
$$
to see that $s_n\circ\ldots s_1=i$, which implies that the number
of $f$'s in $\boldsymbol{s}$ is even. Conversely, if the number of
$f$'s is even, one can take $x_1=0$ and recursively define
$x_{j+1} = s_j(x_j)$, which will give settings $\boldsymbol{x}$
satisfying $x_{j+1}=s_j(x_j)$ for all $j$. If the parity of
$\boldsymbol{x}$ is odd, then flipping all $x_j$ will do the job.
\end{proof}

We are now in position to prove Theorem~\ref{gynitight}. Assuming
$n$ odd, we need to prove that the set of strategies which
saturate $\mathrm{GYNI}_n$ has a linear hull of codimension $1$
within the linear hull of the polytope of CC. To achieve this, we
will show that every strategy can be written as a linear
combination of those strategies saturating GYNI$_n$ together with
the single non-saturating strategy \beq \label{ones} [11\ldots
11]. \eeq To ease the notation, we call two (non-saturating)
strategies $x$ and $y$ \textit{congruent}, written as $x\cong y$,
if their difference $x-y$ is a linear combination of saturating
strategies. Our goal is to show that all non-saturating strategies
are congruent to~(\ref{ones}).

The main tool in the proof is the single-party relation \beq
\label{localsquare} [0] + [1] = [i] + [f] , \eeq which holds since
both the mixture $\frac{1}{2}\left([0]+[1]\right)$ and the mixture
$\frac{1}{2}\left([i]+[f]\right)$ represent pure noise, which is
the local strategy which outputs a random bit independent of the
setting. This relation can be used to express every strategy as a
linear combination of strategies which do not involve $f$. For
example,
$$
[i0f] = [i00] + [i01] - [i0i] .
$$
Thanks to the relation~(\ref{localsquare}), it is enough to prove
the desired congruence only for those non-saturating strategies
which do not contain $f$. Every such strategy contains at least
one $0$ or $1$; for if not, then it would necessarily be equal to
$[i\ldots i]$, which is saturating.

We prove the congruence to~(\ref{ones}) in three steps:

\begin{enumerate}[leftmargin=0cm,itemindent=0cm]
\item \textbf{Claim:} Every non-saturating strategy containing
some $0$ or $1$ is congruent to its evaluation.
In particular, it is congruent to a constant non-saturating strategy.

\textbf{Subproof:} Equation~(\ref{localsquare}) implies
\begin{align}
\begin{split}
\label{rewriting}
[\boldsymbol{s}i0\boldsymbol{t}] + [\boldsymbol{s}f0\boldsymbol{t}] = [\boldsymbol{s}00\boldsymbol{t}] + [\boldsymbol{s}10\boldsymbol{t}] ,\\
[\boldsymbol{s}i1\boldsymbol{t}] + [\boldsymbol{s}f1\boldsymbol{t}] = [\boldsymbol{s}11\boldsymbol{t}] + [\boldsymbol{s}01\boldsymbol{t}] .
\end{split}
\end{align}
for all strings $\boldsymbol{s},\boldsymbol{t}\in\{0,1,i,f\}^*$,
where here and below the notation $X^{*}$ for a set $X$ refers to
the set of strings of any length over $X$. For both equations, the
left-hand side contains exactly one saturating term, and likewise
the right-hand side. Therefore, the non-saturating term on the
left-hand side is congruent to the non-saturating term on the
right-hand side, in which the number of occurrences of $i$ or $f$
is less by one. These congruences can be applied until no further
occurrences of $i$ or $f$ remain. The resulting constant strategy
turns out to be precisely the evaluation of the strategy we
started with.

In the following two steps, we will routinely make use of this congruence.

\item \textbf{Claim:} Every constant non-saturating strategy is congruent to one of the form
\beq \label{onezeroes} [0\ldots 01\ldots 10\ldots 0] \eeq with an
odd number of $1$'s.

\textbf{Subproof:} If $\boldsymbol{s}\in\{0,1\}^*$ is any string
of odd length with an even number of $1$'s, then there is a
congruence \beq \label{basiccong} [10\boldsymbol{s}]\cong
[11\overline{\boldsymbol{s}}], \eeq where
$\overline{\boldsymbol{s}}$ stands for $\boldsymbol{s}$ with all
bits flipped. We show this by writing
$\boldsymbol{s}=[\boldsymbol{t}0]$ or
$\boldsymbol{s}=[\boldsymbol{t}1]$ and considering the two cases
\beq \label{twocases} [10\boldsymbol{t}0] \cong
[11\overline{\boldsymbol{t}}1] ,\qquad [10\boldsymbol{t}1] \cong
[11\overline{\boldsymbol{t}}0] \eeq separately. In the first case,
$\boldsymbol{t}$ has an even number of $1$'s and an even number of
$0$'s. Let $\boldsymbol{u}\in\{i,f\}^{n-2}$ be such that
$[10\boldsymbol{u}]'= [10\boldsymbol{t}0]$; this $\boldsymbol{u}$
contains an even number of $f$'s. Then $[ff\boldsymbol{u}]$
saturates, and we can apply~(\ref{localsquare}) to expand
\begin{align}
\begin{split}
\label{ff} [ff\boldsymbol{\boldsymbol{u}}] = &\phantom{-}
[00\boldsymbol{\boldsymbol{u}}] + [01\boldsymbol{u}] +
[10\boldsymbol{u}] + [11\boldsymbol{u}] \\ & - [0i\boldsymbol{u}]
- [1i\boldsymbol{u}] - [i0\boldsymbol{u}] - [i1\boldsymbol{u}] +
[ii\boldsymbol{u}] .
\end{split}
\end{align}
By the assumption $[10\boldsymbol{u}]'= [10\boldsymbol{t}0]$, we have
\begin{align*}
[00\boldsymbol{u}]' = [00\boldsymbol{t}0], &\quad [01\boldsymbol{u}]' = [01\overline{\boldsymbol{t}}1], \\
[10\boldsymbol{u}]' = [10\boldsymbol{t}0], &\quad [11\boldsymbol{u}]' = [11\overline{\boldsymbol{t}}1], \\
[0i\boldsymbol{u}]' = [00\boldsymbol{t}0], &\quad [1i\boldsymbol{u}]' = [11\overline{\boldsymbol{t}}1], \\
[i0\boldsymbol{u}]' = [00\boldsymbol{t}0], &\quad [i1\boldsymbol{u}]' = [11\overline{\boldsymbol{t}}1].
\end{align*}
Since $\boldsymbol{t}$ contains an even number of $1$'s and an
even number of $0$'s, the only non-saturating strategies on the
right-hand sides of these equations are $[10\boldsymbol{t}0]$ and
$[11\overline{\boldsymbol{t}}1]$. Since each of the first eight
terms on the right-hand side of~(\ref{ff}) is congruent to its
evaluation, we obtain
\begin{align*}
0\phantom{0} \cong &\phantom{-} 0 + \phantom{00}0\phantom{0} +
[10\boldsymbol{t}0] + [11\overline{\boldsymbol{t}}1] \\ & - 0 -
[11\overline{\boldsymbol{t}}1] - \phantom{00}0\phantom{0} -
[11\overline{\boldsymbol{t}}1] + 0
\end{align*}
which proves $[10\boldsymbol{t}0]\cong
[11\overline{\boldsymbol{t}}1]$, as claimed. The second case
of~(\ref{twocases}) works similarly, upon choosing
$\boldsymbol{u}$ such that
$[10\boldsymbol{u}]'=[10\boldsymbol{t}1]$. Then $\boldsymbol{u}$
contains an odd number of $f$'s, so that $[if\boldsymbol{u}]$ is
saturating. Using the expansion
\begin{align}
\begin{split}
[if\boldsymbol{u}] = &\phantom{-} [00\boldsymbol{u}] +
[01\boldsymbol{u}] + [10\boldsymbol{u}] + [11\boldsymbol{u}] \\ &
- [0i\boldsymbol{u}] - [1i\boldsymbol{u}] - [f0\boldsymbol{u}] -
[f1\boldsymbol{u}] + [fi\boldsymbol{u}]
\end{split}
\end{align}
and applying reasoning analogous to the previous case shows $[10\boldsymbol{t}1]\cong [11\overline{\boldsymbol{t}}0]$.

Due to cyclic symmetry,~(\ref{basiccong}) actually implies
$$
[\boldsymbol{s}10\boldsymbol{t}]\cong [\overline{\boldsymbol{s}}11\overline{\boldsymbol{t}}] \quad\forall \boldsymbol{s},\boldsymbol{t}\in\{0,1\}^*.
$$
For any $\boldsymbol{r},\boldsymbol{u}\in\{0,1\}^*$ of appropriate
length and number of $1$'s, applying this transformation yields
the congruences
\begin{align*}
[\boldsymbol{r}100\boldsymbol{u}]\cong[\overline{\boldsymbol{r}}111\overline{\boldsymbol{u}}]\cong [\boldsymbol{r}010\boldsymbol{u}], \\
[\boldsymbol{r}101\boldsymbol{u}]\cong[\overline{\boldsymbol{r}}110\overline{\boldsymbol{u}}]\cong [\boldsymbol{r}011\boldsymbol{u}].
\end{align*}
In both cases, the right-hand side is equal to the left-hand side,
except for one $1$, which has been transported one position to the
right. Repeated application of these congruences transforms any
constant non-saturating strategy into one of the
form~(\ref{onezeroes}).

\item \textbf{Claim:} Every strategy of the form~(\ref{onezeroes}) is congruent to the strategy $[1\ldots 1]$.

\textbf{Subproof:} To see this, we consider the saturating strategy
$$
[i\ldots ifi\ldots ifi\ldots i]
$$
where the number of $i$'s in the middle is assumed to be one less
than the number of $1$'s in~(\ref{onezeroes}). Expanding the two
$f$'s by~(\ref{localsquare}) and keeping only the non-saturating
terms gives
\begin{align*}
0\cong & \phantom{-} [i\ldots i1i\ldots i0i\ldots 1] + [i\ldots i1i\ldots i1i\ldots i]\\
 & - [i\ldots i1i\ldots iii\ldots i] - [i\ldots iii\ldots i1i\ldots i].
\end{align*}
The first term evaluates to~(\ref{onezeroes}), while the others evaluate to $1\ldots 1$, thereby proving the desired congruence.
\end{enumerate}

In conclusion, any non-saturating strategy can be written as a
linear combination of saturating strategies and $[1\ldots 1]$ by
first expanding all $f$'s by~(\ref{localsquare}) and then applying
steps $1$ to $3$ to each of the resulting terms.


\end{document}